\documentclass[12pt,a4paper]{article}
\usepackage[T1]{fontenc}
\usepackage[english]{babel}
\usepackage{graphicx}
\usepackage{subfigure}
\usepackage{geometry}
\usepackage{amsmath}
\usepackage{amssymb}
\usepackage{amsthm}
\usepackage{authblk}
\usepackage{mathtools}
\usepackage{caption}
\usepackage{ragged2e}
\usepackage{cite}
\usepackage{placeins}
\usepackage{indentfirst}
\usepackage{setspace}
\usepackage{enumerate}
\usepackage{enumitem}
\usepackage{varioref}
\usepackage{booktabs,tabularx}
\usepackage{color}
\usepackage{epstopdf}
\usepackage{longtable}
\usepackage{array}
\makeatletter \oddsidemargin  -.1in \evensidemargin -.1in
\title{\textbf{ Statistical inference of a competing risks model based on improved adaptive type-II progressive censored data under Gompertz lifetime distribution}}
\author[1]{Amlan \textbf{Dey}\thanks{Email: amlandey.mtb@gmail.com}}
\author[2]{Subhankar \textbf{Dutta}\thanks{Email: subhankar.dta@gmail.com}}
\author[3]{Suchandan \textbf{Kayal}\thanks{Email:  kayals@nitrkl.ac.in, suchandan.kayal@gmail.com}}
\affil[1,3]{Department of Mathematics, National Institute of
	        	Technology Rourkela, Rourkela-769008, India}
\affil[2]{Department of Mathematics, Bioinformatics, and Computer Applications, Maulana Azad National Institute of Technology Bhopal, Bhopal-462003, India}
\date{}
\newcommand{\bea}{\begin{eqnarray}}
	\newcommand{\eea}{\end{eqnarray}}

\newcommand{\bee}{\begin{eqnarray*}}
	\newcommand{\eee}{\end{eqnarray*}}

\newtheorem{thm}{Theorem}[section]

\newtheorem{lemma}{Lemma}[section]

\numberwithin{equation}{section}

\begin{document}
  \maketitle
  \begin{abstract}
  	This paper studies survival of certain bird species (Zebra Finches) under different food availability conditions. A competing risks model is studied under improved adaptive type-II progressive censoring scheme (IAT-II PCS) referring to this phenomena. Two independent competing causes of failure are considered where lifetime of these failures are assumed to follow Gompertz distribution with unknown scale and shape parameters. Maximum likelihood estimators (MLEs) of the unknown parameters are derived. It is established that they exist uniquely. Asymptotic confidence intervals (ACIs) are also constructed using asymptotic normality property of the MLE.  Bayes estimates are obtained with respect to both non-informative and informative priors under different loss functions. Highest posterior density (HPD) credible intervals are calculated. A Monte Carlo simulation study is conducted to compare the performance of the proposed estimates. Three optimality criteria are studied to obtain the optimal censoring scheme. Finally, a real life data set is analyzed for further illustrations.
  \end{abstract}
  \section{Introduction}
  Food accessibility plays a crucial role in population structure and life-history adaptation in natural communities since it is believed that increased accessibility to food increases fertility and mortality. Food availability correlates with various ecological parameters, therefore it remains ambiguous precisely how accessibility to food impacts the ability to reproduce and survive. Therefore, a surge in accessibility to food may have a major impact on mortality by altering the rate of predation with a minor impact on the rate of famine. Since animals are frequently more exposed when hunting and high food supply enables a decline in the amount of time invested scavenging a surge in food accessibility, for instance, is likely to lower the danger of malnutrition but also has an impact on the presence of predators. Therefore, a rise in accessibility to food may have a primary impact on survival by altering the extent
  of predatory behavior, with a minor impact on the rate of famine. Individual variability
  is influenced by environmental variables throughout all age groups but it is believed that
  environmental circumstances during development have a particularly significant role in shaping
  variations in longevity and, more broadly, adult health. Such different cases effect the
  mortality of animals which factors can be considered as risks factors in survival analysis.

  In survival analysis, experimental items may fail for multiple reasons. The experimenter documents both the failure time of the items and the specific causes of these failures. These causes compete to result in item's failure, a scenario called competing risks in the literature and it has been studied quite extensively by several researchers, for example, Kalbfleish and Prentice \cite{kalbfleisch2002statistical} and Lawless \cite{lawless2011statistical}. Data with competing risks include both the failure time and indicators denoting the causes of failure. Two common methods for analyzing this type of data are latent failure time model proposed by Cox \cite{cox1959analysis} and cause-specific hazard rate model proposed by Prentice et. al \cite{prentice1978analysis}. Competing risks studies have broad applications in reliability engineering and biomedical research. For example, a heart disease patient might die from various causes other than heart failure. Usually, the treatment failure has two events, say relapse or mortality due to treatment procedure in the studies of stem cell transplantation. These two events yield competing risks data although it is of interest to consider the time to the first event. In a study that examines the probability of staphylococcus infection during hospital admissions, treatment failure or death may occur from two events such as the event of interest is staphylococcus infection and the competing event is mortality before the said infection or other kind of infections. In traditional survival analysis, competing events are typically handled as right-censored observations; therefore, a competing risks framework is required when such events are present. For reference, one may go through Crowder \cite{crowder2012multivariate} for studies on various applications of competing risks model.
  
  Sometimes experimental items are withdrawn from the experiment before the failure-time arises due to multiple reasons such as enhanced lifetime of the experimental product, cost and time limitations. This results in censoring the available data. In life-testing experiment, censoring is implemented to reduce time consumption and lower the amount of lost items. Survival studies have made extensive use of type-I and type-II conventional censoring schemes. Chandrashekhar et al. \cite{chandrasekar2004exact} introduced generalized version of mixture of these two censoring schemes, which is known as generalized type-I and type-II hybrid censoring schemes. The primary disadvantage of conventional censoring schemes is their inability to remove experimental items before the termination of the experiment. To address this issue, Cohen \cite{cohen1963progressively} introduced progressive censoring scheme which allows us to remove survival items at various stages of an experiment. Significant advancements in statistical inference for certain lifetime distributions using these censoring schemes have been made in recent years. See, for example, H.K.T. Ng \cite{ng2005parameter}, Balakrishnan et. al \cite{balakrishnan2003point} and Rastogi and Tripathi \cite{rastogi2012estimating}. It is important to note that experiments using progressive type-II censoring schemes can sometimes take longer to reach a predetermined number of failures. To overcome this drawback,  Kundu and Joarder \cite{kundu2006analysis} introduced type-II progressive hybrid censoring scheme. Subsequently, many researchers have explored this scheme. See, for example, Kundu et. al \cite{kundu2009type},  Hemmati and Khorram \cite{hemmati2013statistical} and Dutta and Kayal \cite{dutta2022estimation}. In this censoring scheme, the experimental duration is predetermined, while the effective sample size is random. Occasionally, the effective sample size drops to zero, leading to a decline in the efficiency of statistical inference. To overcome this drawback, Ng et. al \cite{ng2009statistical} introduced the adaptive type-II progressive censoring scheme (AT-II PCS). One may refer to Cramer and LLiopoulos \cite{cramer2010adaptive},  Ye et. al \cite{ye2014statistical} and Almetwally et. al \cite{almetwally2020maximum} for more details on AT-II PCS. In AT-II PCS, the experiment concludes after obtaining a predetermined number of failure-time data, which can result in a lengthy experiment duration. Ng et. al \cite{ng2009statistical} highlighted that the AT-II PCS is only efficient for inferences when experiment duration is not a limiting factor. For highly reliable items, this approach may lead to excessively longer periods, making it unsuitable in situations where keeping the experiment time within a practical limit is crucial. To overcome this drawback, Yan et. al \cite{yan2021statistical} introduced improved adaptive type-II (IAT-II) progressive censoring scheme which is particularly beneficial when managing the experiment duration is crucial. This improved method extends several censoring techniques, including T-II PCS and AT-II PCS, while guaranteeing that experiments are completed within a set time frame, thereby addressing the problem of prolonged experiment duration. This paper analyzes the food availability data under IAT-II PCS where IAT-II PCS data is generated from the real data under competing risks framework. We refer Dutta and Kayal \cite{dutta2023inference} for more details on IAT-II PCS which is described below.
  

  Consider an experiment, where $n$ experimental items are taken. At the start of the experiment, we fix the number of failure, say $m$ and the progressive censoring scheme $(r_1,\cdots,r_m)$, where $r_m=n-m-\sum_{j=1}^{m-1}r_j$ and $r_j\geq 0$, for $j=1,\cdots,m$. We denote $X_{j:m:n}$ as the random lifetime of the $j$th failure item. Also, we take two time thresholds $t_1$ and $t_2$ before starting the test such that $0<t_1<t_2$. In this context, the first threshold, $t_{1}$ serves as a warning for the experiment duration. The second threshold, $t_2$ indicates the maximum duration for which the test will be conducted, regardless of whether the experimenter achieves $m$ failures. Here, we present three cases (described below) in Figure $1$ for better visualization.  
  \begin{align}
  	\nonumber \mbox{Case-I}~~~~& X_{m:m:n} < t_{1}<t_{2},\\
  	\nonumber \mbox{Case-II}~~~& X_{k_{1}:m:n}<t_{1}<X_{k_{1}+1:m:n}<\cdots<X_{m:m:n}<t_{2},\\
  	\nonumber \mbox{Case-III}~~& X_{k_{2}:m:n}<t_{2}<X_{k_{2}+1:m:n}<\cdots<X_{m:m:n}.
  \end{align}
  Under the above three cases, we get different sets of failure-times as 
  \begin{align}
  	\nonumber \mbox{Case-I}~~~& \{X_{1:m:n}<\cdots<X_{m:m:n}\},\\
  	\nonumber \mbox{Case-II}~~& \{X_{1:m:n}<\cdots<X_{k_{1}:m:n}<\cdots<X_{m:m:n}\},\\
  	\nonumber \mbox{Case-III}~& \{X_{1:m:n}<\cdots<X_{k_{1}:m:n}<\cdots<X_{k_{2}:m:n}\}.
  \end{align}
  We note that the experiment termination time is $t=min\{X_{m:m:n},t_{2}\}$ under IAT-II PCS.\\
  \begin{figure}[h!]
   \begin{center}
  	\includegraphics[height=3.6in]{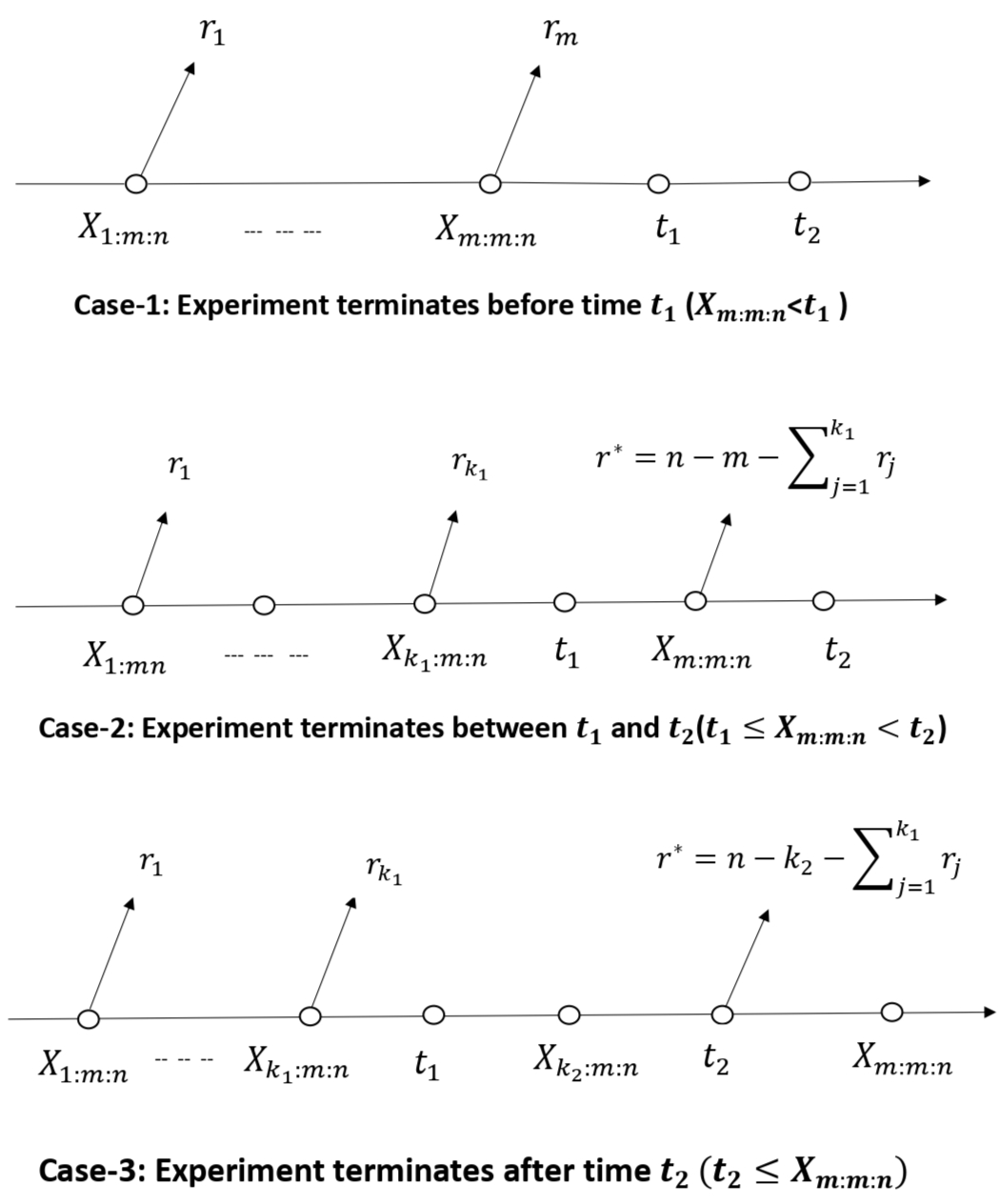}
  	\caption{ Illustrated diagram of the IAT-II PCS.}
  	\label{fig:fig1}
   \end{center}
  \end{figure}
  Numerous researchers have examined competing risks data under different censoring schemes. Kundu and Basu \cite{kundu2000analysis} studied Weibull competing risks in the presence of incomplete data among several groups. Kundu and Joarder \cite{kundu2006analysis} used type-II progressive hybrid censored data to analyze the competing risks model. Sarhan \cite{sarhan2007analysis} used a generalized exponential distribution to examine the competing risks model with censored and incomplete data. Cramer \cite{cramer2011progressively} used the Lomax distribution to examine progressively censored competing risks model. Chacko and Mohan \cite{chacko2019bayesian} studied Bayesian analysis of type-II progressively censored competing risks model with binomial removals.
  Ahmed et. al \cite{ahmed2020inference} studied the inference of progressive type-II competing risks model based on Chen distribution. Mahmoud et. al \cite{mahmoud2021inference} estimated parameters of generalized inverted exponential distribution based on progressively type-I censored competing risks data.  Nassr et. al \cite{nassr2021statistical} examined extended Weibull distributions under adaptive type-II (AT-II) progressively censored competing risks model.     
  
  As far as we know, till now, very few study has been done based on the improved adaptive type-II progressive censored competing risks model in the literature. We have utilized a recently introduced censoring scheme (IAT-II PCS) to examine competing risks data, where the lifetimes of the competing causes of failure follow Gompertz distributions with different shape and scale parameters. This distribution displays increasing failure rate models which is very adaptable for fitting with various survival and mortality data (See Figure \ref{fig:fig2}). The motivations for the Gompertz distribution is provided in upcoming section. Classical and Bayes estimates of the unknown model parameters of Gompertz distribution under IAT-II progressive censored sample and asymptotic confidence intervals are derived. The HPD credible intervals are also obtained. We have studied three different optimality criteria to determine the optimal censoring scheme among the chosen schemes. Finally, we have analyzed a real life dataset.  
  
  The paper is divided in certain sections as follows. In Section $2$, the competing risks model and the corresponding likelihood function under IAT-II PCS are introduced. In section $3$, the MLEs are derived and the existence and uniqueness properties of MLEs are studied. Also, asymptotic confidence intervals are constructed in this section. In Section $4$, a method to determine Bayes estimates is demonstrated for three different loss functions and HPD credible intervals for those estimates are discussed. In Section $5$, a hypothesis testing problem is formulated using likelihood ratio statistic. In Section $6$, a Monte Carlo simulation study is conducted to compare the performance of the proposed estimates and some inferences are made. Three different optimality criteria have been introduced in Section $7$. In Section $8$, a real-life dataset is analyzed for further illustrations of the proposed methods. Finally, in Section $9$, a conclusion has been drawn.

  \section{Model description and the likelihood function}
  In this section, the latent failure time model proposed by Cox \cite{cox1959analysis} is considered to study the competing risks data. It is assumed that two competing causes of failures are mutually independent. Let $X_{j}$ denote the lifetime of the $j$th item. Then, 
  \begin{align}
  	\nonumber X_{j}= min\{X_{1j},X_{2j}\},
  \end{align}
  where $X_{ij}$ $(i=1,2)$ is latent failure time of the $j$th item when the $i$th cause of failure occurs. Here, $X_{1j}$ and $X_{2j}$ are independent latent failure times. Consider ${\delta}_{j}=1$ and ${\delta}_{j}=0$ when the failure occurs due to cause 1 and cause 2, respectively. Then, $N_{1}=\sum_{j=1}^{N}I({\delta}_{j}=1)$ is the number of failures due to cause 1 and  $N_{2}=\sum_{j=1}^{N}I({\delta}_{j}=0)$ is the number of failures due to cause 2, where $N_{1}+N_{2}=N$.
  The cumulative distribution function (CDF), probability density function (PDF), survival function and hazard rate function (HRF) of the latent failure times $X_{ij}$, for $j=1,\cdots,n$ and $i=1,2$ are denoted by $F_{i}(\cdot)$, $f_{i}(\cdot)$, $S_{i}(\cdot)$ and $h_{i}(\cdot),$ respectively. Under these settings, the likelihood function is written as (see also Dutta and Kayal \cite{dutta2023inference})
  \begin{align}
  	 L({\alpha}_{1},{\alpha}_{2},\beta) \propto \bigg[f_{1}(x_{j}){S_{2}}(x_{j})\bigg]^{N_{1}} \bigg[f_{2}(x_{j}){S_{1}}(x_{j})\bigg]^{N_{2}} \prod_{j=1}^{N}\bigg[{S_{1}}(x_{j}){S_{2}}(x_{j})\bigg]^{r_{j}} \bigg[{S_{1}}(t){S_{2}}(t)\bigg]^{r^{*}}, \label{2.1}
  \end{align}
  where $r^*$ is described later and $\alpha_{1},{\alpha}_{2},\beta$ are the model parameters associated with the model under study. Using $f_{i}(x)=h_{i}(x)S_{i}(x)$, \eqref{2.1} changes into
  \begin{equation}
  	L({\alpha}_{1},{\alpha}_{2},\beta) \propto \bigg[h_{1}(x_{j})\bigg]^{N_{1}} \bigg[h_{2}(x_{j})\bigg]^{N_{2}} \prod_{j=1}^{N}\bigg[{S_{1}}(x_{j}){S_{2}}(x_{j})\bigg]^{1+r_{j}}\bigg[{S_{1}}(t){S_{2}}(t)\bigg]^{r^{*}},\label{2.2}
  \end{equation}
  where $x_{j}$ represents $j$th failure time $x_{j:m:n}$. Note that for case-I, $N=m$ and $r^{*}= 0$; for case-II, $N=k_{1}$ and $r^{*}= n-m-\sum_{j=1}^{k_{1}}r_{j}$, and for case-III, $N=k_{2}$ and $r^{*}= n-k_{2}-\sum_{j=1}^{k_{1}}r_{j}$. Here, we have assumed that lifetime of the competing causes of failure follows Gompertz distribution with PDF, CDF and hazard rate function (HRF), as follows
  \begin{equation}
  	f_{i}(x)={\alpha}_{i}e^{\beta x}e^{-\frac{\alpha_{i}}{\beta}(e^{\beta x}-1)},~ F_{i}(x)=1-e^{-\frac{\alpha_{i}}{\beta}(e^{\beta x}-1)}~~\mbox{and}~~h_{i}(x)={\alpha}_{i}e^{\beta x},\label{2.3}
  \end{equation}
  where $\beta,{\alpha}_{i}, x>0,$ for $i=1,2$.
  
  The PDF, CDF and HRF of the Gompertz distribution are plotted in Figure \ref{fig:fig2}. A per the Figure \ref{fig:fig2}, PDF of Gompertz distribution can be skewed to right and left by adjusting the values of $\alpha_{i}$ and $\beta$. The HRF of this distribution is increasing which implies that the failure rate of a component increases with respect to time. Due to this property, Gompertz distribution serves as a classical model for representing mortality patterns, especially in human and animal populations, where the risk of death tends to increase with age. Also in reliability studies, it helps to understand the failure rates of components or systems over time, particularly when the failure rate increases with time. Dey et. al \cite{dey2018statistical} studied the properties and estimation methods of Gompertz distribution. We refer to Wu et. al \cite{wu2006mle}, Soliman et. al\cite{soliman2015bayesian} and Wu et. al \cite{wu2017statistical} for more details on Gompertz distribution under different progressive censoring schemes.
  \begin{figure}[ht!]
  	\centering
  	\subfigure[]{\includegraphics[width=0.32\textwidth,height=0.31\textwidth]{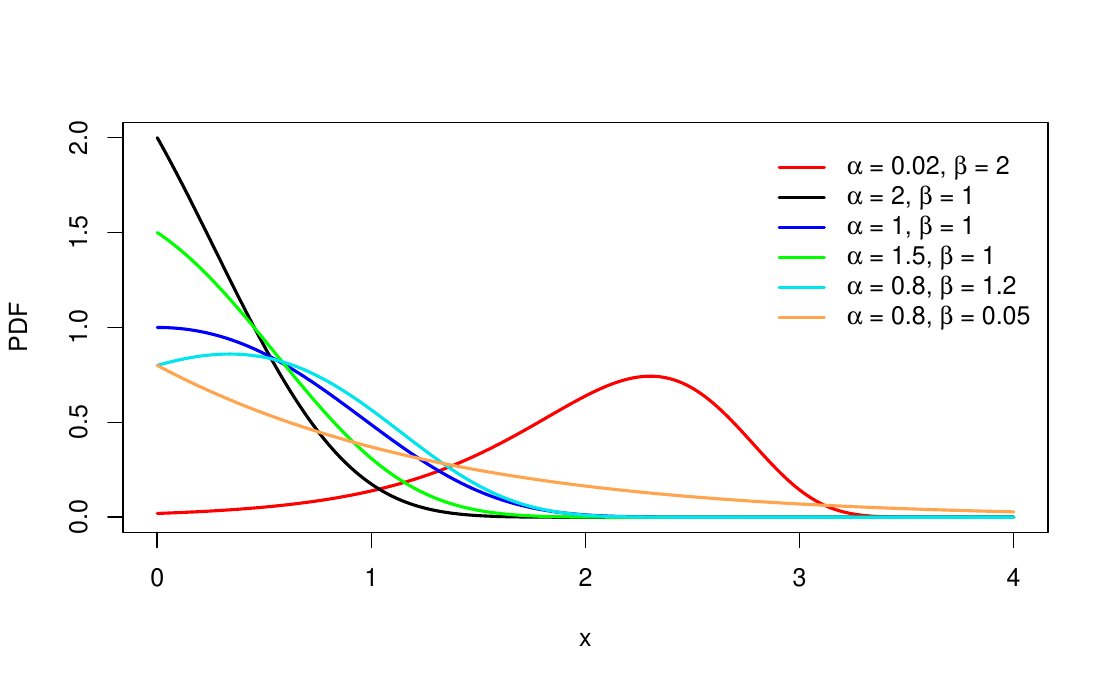}}
  	\subfigure[]{\includegraphics[width=0.32\textwidth,height=0.31\textwidth]{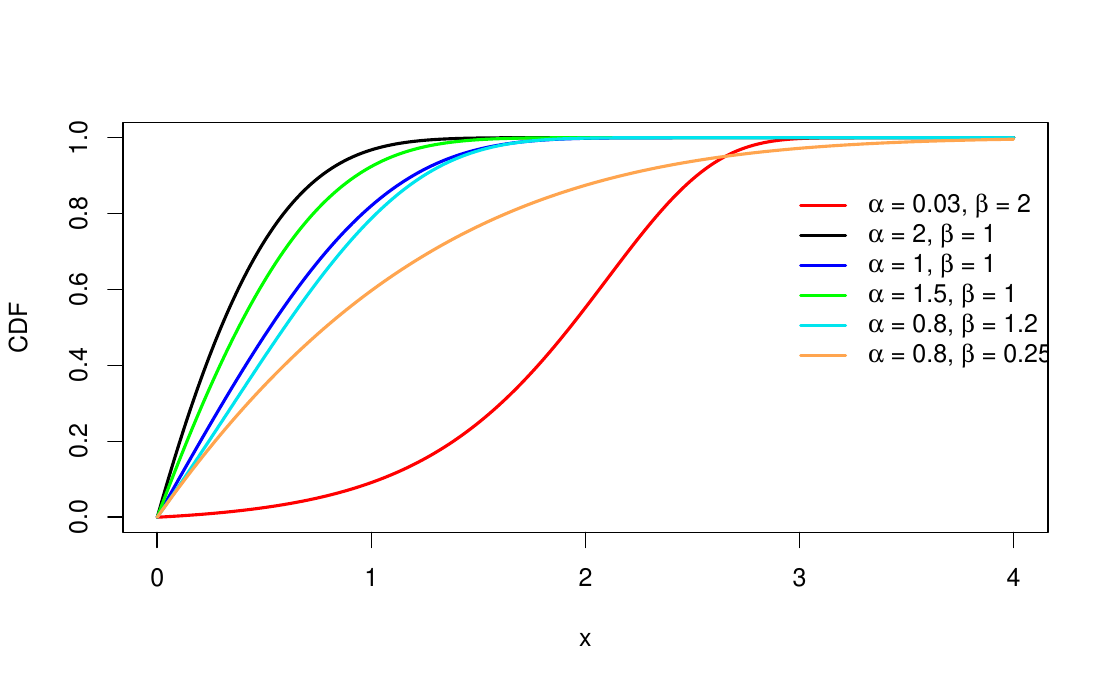}}
  	\subfigure[]{\includegraphics[width=0.32\textwidth,height=0.31\textwidth]{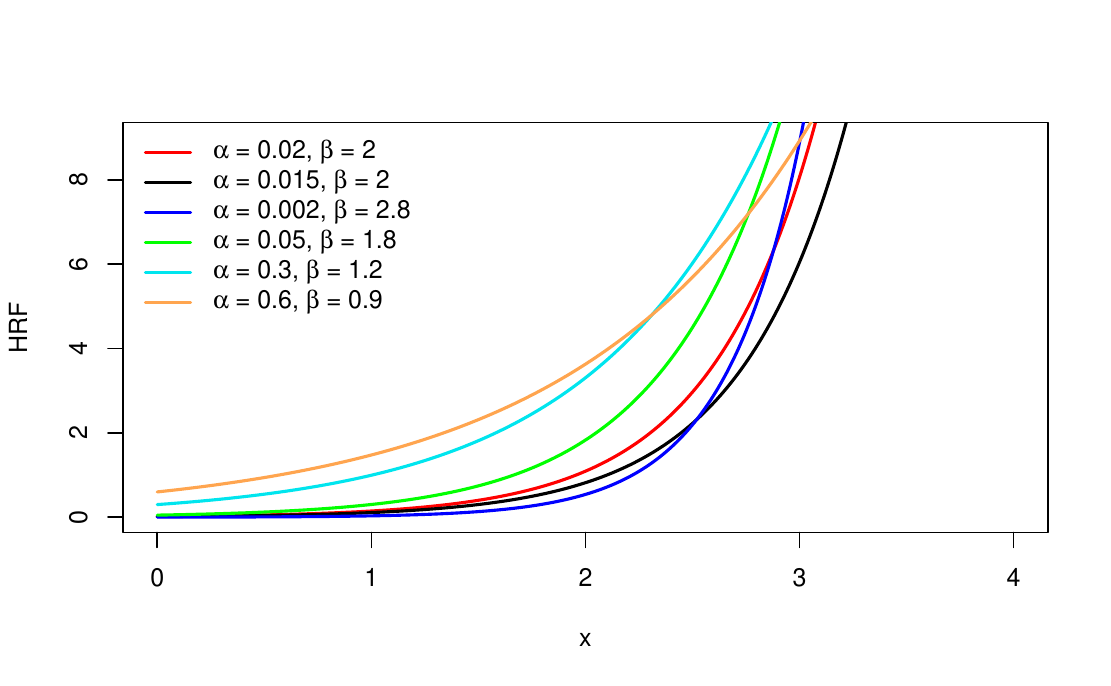}}
  	\caption{(a) PDF, (b) CDF, (c) HRF plots for different scale and shape parameters of the Gompertz distribution.}
  	\label{fig:fig2}
  \end{figure}
  
  \section{Maximum likelihood and interval estimations}
  In this section, we derive MLEs of ${\alpha}_{1}$, ${\alpha}_{2}$ and $\beta$. From \eqref{2.2} and \eqref{2.3}, the likelihood function is obtained as 
  \begin{equation}
  	L({\alpha}_{1},{\alpha}_{2},\beta) \propto~ {{\alpha}_{1}}^{N_{1}} {{\alpha}_{2}}^{N_{2}} \prod_{j=1}^{N} e^{\beta x_{j}} \Big[e^{-\frac{(\alpha_{1}+\alpha_{2})}{\beta}(1+r_{j})(e^{\beta x_{j}}-1)}\Big] e^{-\frac{(\alpha_{1}+\alpha_{2})}{\beta}t^*},\label{3.1}
  \end{equation}
  where $t^{*}=(e^{\beta t}-1)r^{*}$. Now, we can formulate the log-likelihood function from the likelihood function as
  \begin{equation}
  	\log L({\alpha}_{1},{\alpha}_{2},\beta) \propto N_{1} \log \alpha_{1} +  N_{2} \log \alpha_{2} + \sum_{j=1}^{N} \beta x_{j}-\frac{(\alpha_{1}+\alpha_{2})}{\beta}\bigg[ \sum_{j=1}^{N}(1+r_{j})(e^{\beta x_{j}}-1)+t^*\bigg].\label{3.2}
  \end{equation}
  After differentiating \eqref{3.2} with respect to ${\alpha}_{1}$, ${\alpha}_{2}$, $\beta$, and then equating them to zero, we obtain the likelihood equations, given by
  \begin{equation}
  	\frac{\partial l}{\partial {\alpha}_{1}}= \frac{N_{1}}{\alpha_{1}}-\frac{1}{\beta}\bigg[\sum_{j=1}^{N} (1+r_{j})(e^{\beta x_{j}}-1)+t^*\bigg],\label{3.3}
  \end{equation}
  \begin{equation}
  	\frac{\partial l}{\partial {\alpha}_{2}}= \frac{N_{1}}{\alpha_{2}}-\frac{1}{\beta}\bigg[\sum_{j=1}^{N} (1+r_{j})(e^{\beta x_{j}}-1)+t^*\bigg],\label{3.4}
  \end{equation}
  \begin{equation}
  	\frac{\partial l}{\partial \beta}= \sum_{j=1}^{N} x_{j} + \frac{(\alpha_{1}+\alpha_{2})}{\beta}\bigg[\frac{\sum_{j=1}^{N} (1+r_{j})(e^{\beta x_{j}}-1)+t^*}{\beta}-\sum_{j=1}^{N} (1+r_{j})x_{j}e^{\beta x_{j}}+r^*te^{\beta t}\bigg], \label{3.5}
  \end{equation}
  where $l= \log L({\alpha}_{1},{\alpha}_{2},\beta)$. From \eqref{3.3} and \eqref{3.4}, $\alpha_1$ and $\alpha_2$, for any given $\beta$, can be derived:
  \begin{equation}
  	{\alpha}_{1}= \frac{N_{1}\beta}{A(x,\beta)}~~~\mbox{and}~~~{\alpha}_{2}= \frac{N_{2}\beta}{A(x,\beta)}, \label{3.6}
  \end{equation}
  where $A(x,\beta)= \sum_{j=1}^{N} (r_{j}+1)(e^{\beta x_{j}}-1) +t^{*} $.\\
   
  From likelihood equation in \eqref{3.5}, the MLE of $\beta$, denoted by $\widehat{\beta}$, can be obtained by solving the following non-linear equation
  \begin{equation}
  	\frac{N}{\beta}+ \sum_{j=1}^{N} x_{j} -\frac{NB(x,\beta)}{A(x,\beta)}=0, \label{3.7}
  \end{equation}
  where $B(x,\beta)=\sum_{j=1}^{N}(1+r_{j})x_{j}e^{\beta x_{j}}+r^{*}te^{\beta t}$.\\
  
  From \eqref{3.6} and \eqref{3.7}, the MLEs of $\alpha_{1}$ and $\alpha_{2}$, denoted by $\widehat{{\alpha}_{1}}$ and $\widehat{{\alpha}_{2}},$ respectively, can be obtained, which are given by 
  \begin{equation}
  	\widehat{{\alpha}_{1}}= \frac{N_{1}\widehat{\beta}}{A(x,\widehat{\beta})}~~\mbox{and}~~\widehat{{\alpha}_{2}}= \frac{N_{2}\widehat{\beta}}{A(x,\widehat{\beta})}. \label{3.8}
  \end{equation}
  
  Since $\widehat{\beta}$ is the MLE of $\beta$, it is evident from \eqref{3.7} that $\widehat{\beta}$ satisfies the following fixed type equation
  \begin{equation}
  	\beta=g(\beta)=N\bigg[\frac{NB(x,\beta)}{A(x,\beta)}-\sum_{i=1}^{N} x_{i}\bigg]^{-1}. \label{3.9}
  \end{equation}
  Using fixed point iteration method proposed by Kundu \cite{kundu2007hybrid}, we can solve \eqref{3.10} for $\beta$.
  The following algorithm is useful to obtain $\widehat{\beta}$.\\
  ------------------------------------------------------------------------------------------------------------------------
  \textbf{Algorithm I}\\
  ------------------------------------------------------------------------------------------------------------------------
  \textbf{Step-1:} Take an initial guess  $\beta^{(0)}$ of $\beta$.\\
  \textbf{Step-2:} Obtain $\beta^{(1)}=g(\beta^{(0)})$.\\
  \textbf{Step-3:} Obtain $\beta^{(k)}$ from $\beta^{(k)}=g(\beta^{(k-1)})$.\\
  \textbf{Step-4:} Set $k=k+1$.\\
  \textbf{Step-5:} Terminate the process if  $|\beta^{(k+1)}-\beta^{(k)}|<\epsilon$, for a preferred tolerance limit $\epsilon$.\\
  \textbf{Step-6:} Take final result as $\widehat{\beta}$.\\
  ------------------------------------------------------------------------------------------------------------------------ 
  After obtaining $\widehat{\beta}$, the MLEs of $\alpha_{1}$ and $\alpha_{2}$ can be easily obtained from \eqref{3.8}, respectively denoted as $\widehat{{\alpha}_{1}}$ and $\widehat{{\alpha}_{2}}$.
  
  \subsection{Existence and uniqueness of MLEs}
 Here, we present a lemma, which will be useful to establish the existence and uniqueness of the MLEs, stated in the next theorem.
  \begin{lemma}
  	Let $x_{1}\leq\cdots\leq x_{j}\leq x_{j+1}\leq\cdots\leq x_{N}$ for $1\leq j\leq N $. Then, for all $s_{j}>0$, $1\leq j\leq N$, we have 
  	\begin{enumerate}[label=(\roman*)]
  		\item \begin{align}\nonumber
  		\lim_{\beta\rightarrow 0} \biggl\{\frac{1}{\beta}-\frac{\sum_{j=1}^{N} s_{j}x_{j}e^{\beta x_{j}}+tr^{*}e^{\beta t}}{\sum_{j=1}^{N}s_{j}(e^{\beta x_{j}}-1)+t^*}\biggr\}=-\frac{\sum_{j=1}^{N}s_{j}x^{2}_{j}+r^*t^{2}}{2\big(\sum_{j=1}^{N}s_{j}x_{j}+r^*t\big)};
  		\end{align}
  		\item 
  		
  			 \begin{eqnarray*}
  				\lim_{\beta\rightarrow \infty} \frac{\sum_{j=1}^{N}s_{j}x_{j}e^{\beta x_{j}}+r^*te^{\beta t}}{\sum_{j=1}^{N}s_{j}(e^{\beta x_{j}}-1)+t^*} =
  			\left\{
  			\begin{array}{ll}
  				x_{N},  ~~\mbox{for case I}\\
  				t,~~~~\mbox{for case II and case III}.
  			\end{array}
  			\right.
  		\end{eqnarray*}		
  	\end{enumerate}
  \end{lemma}
  \begin{proof}
  	\begin{enumerate}[label=(\roman*)]
  		\item For convenience, denote $T_{1}(\beta)=\frac{1}{\beta}-\frac{\sum_{j=1}^{N} s_{j}x_{j}e^{\beta x_{j}}+tr^{*}e^{\beta t}}{\sum_{j=1}^{N}s_{j}(e^{\beta x_{j}}-1)+t^*}$. We have
  		\begin{align}\nonumber
  		\begin{aligned}
  		\lim_{\beta\rightarrow 0} T_{1}(\beta)&=\lim_{\beta\rightarrow 0} \biggl\{\frac{\sum_{j=1}^{N}s_{j}(e^{\beta x_{j}}-1)+t^*-\beta\sum_{j=1}^{N} s_{j}x_{j}e^{\beta x_{j}}+\beta tr^{*}e^{\beta t}}{\beta\sum_{j=1}^{N}s_{j}(e^{\beta x_{j}}-1)+\beta t^*}\biggr\}\left(\frac{0}{0}~\mbox{form}\right)\\
  		&=-\lim_{\beta\rightarrow 0} \biggl\{\frac{\sum_{j=1}^{N} \beta s_{j}x^2_{j}e^{\beta x_{j}}+\beta t^2r^*e^{\beta t}}{\sum_{j=1}^{N}s_{j}\{(1+\beta x_{j})e^{\beta x_{j}}-1\}+r^*\{(1+\beta t)e^{\beta t}-1\}}\biggr\}\left(\frac{\infty}{\infty}~\mbox{form}\right)\\
  		&=-\lim_{\beta\rightarrow 0} \biggl\{\frac{\sum_{j=1}^{N}s_{j}x^2_{j}(1+\beta x_{j})e^{\beta x_{j}}+r^*t^2(1+\beta t)e^{\beta t}}{\sum_{j=1}^{N}s_{j}x_{j}(2+\beta x_{j})e^{\beta x_{j}}+r^*t(2+\beta t)e^{\beta t}}\biggr\}\\
  		&=-\frac{\sum_{j=1}^{N}s_{j}x^{2}_{j}+r^*t^{2}}{2\big(\sum_{j=1}^{N}s_{j}x_{j}+r^*t\big)},\\
  		\end{aligned}
  		\end{align}
  		where the second equality and third equality follow from L'Hospital's rule. This completes the proof of the first part.
  		\item We have
  		\begin{align}\nonumber
  		\begin{aligned}
  		\lim_{\beta\rightarrow \infty} \biggl\{\frac{\sum_{j=1}^{N}s_{j}x_{j}e^{\beta x_{j}}+r^*te^{\beta t}}{\sum_{j=1}^{N}s_{j}(e^{\beta x_{j}}-1)+t^*}\biggr\} \left(\frac{\infty}{\infty}~\mbox{form}\right) &= \lim_{\beta\rightarrow \infty} \biggl\{\frac{\sum_{j=1}^{N} s_{j}x^2_{j}e^{\beta x_{j}}+t^2r^*e^{\beta t}}{\sum_{j=1}^{N} s_{j}x_{j}e^{\beta x_{j}}+tr^{*}e^{\beta t}} \biggr\}\\
  		&=\lim_{\beta\rightarrow \infty} \biggl\{\frac{\sum_{j=1}^{N} s_{j}x^2_{j}e^{\beta(x_{j}-t)}+r^*t^2}{\sum_{j=1}^{N} s_{j}x_{j}e^{\beta(x_{j}-t)}+r^*t}\biggr\}\\
  		&=\left\{
  		\begin{array}{ll}
  			x_{N},  ~~\mbox{for case I}\\
  			t,~~~~\mbox{for case II and case III}.
  		\end{array}
  		\right.\\
  		\end{aligned}
  		\end{align}
  	\end{enumerate}
  	Hence, the proof follows.
  \end{proof}
  \begin{thm}
  	Let $x_{1}\leq\cdots\leq x_{j}\leq x_{j+1}\leq\cdots\leq x_{N}$ for  $1\leq j\leq N $. Then, the MLEs $\widehat{{\alpha}_{i}}$ and $\widehat{\beta}$ of the parameters $\alpha_{i}$ and $\beta$ exist and are unique, with
  	\begin{align}\nonumber
  	\widehat{{\alpha}_{i}}= \frac{N_{i}\widehat{\beta}}{A(x,\widehat{\beta})},~~\mbox{for}~~i=1,2
  	\end{align}
  	and $\widehat{\beta}$ is the solution of the following nonlinear equation:
  	\begin{align}\nonumber
  	G(\beta)=\frac{N}{\beta}+ \sum_{j=1}^{N} x_{j} -\frac{NB(x,\beta)}{A(x,\beta)}=0,
  	\end{align}
  	if
  	\begin{itemize}
  		\item[(i)] $\sum_{j=1}^{N}(1+r_{j})t(e^{\beta x_{j}}-1)>2\sum_{j=1}^{N}(1+r_{j})x_{j}e^{\beta x_{j}}$ and 
  		\item[(ii)] $2\sum_{j=1}^{N}x_{j}\big(\sum_{j=1}^{N}(1+r_{j})x_{j}+r^*t\big)>N\big(\sum_{j=1}^{N}(1+r_{j})x^2_{j}+r^*t^2\big),$
  		\end{itemize}
  	where $A(x,\beta)= \sum_{j=1}^{N} (r_{j}+1)(e^{\beta x_{j}}-1) +t^{*}$ and $B(x,\beta)=\sum_{j=1}^{N}(1+r_{j})x_{j}e^{\beta x_{j}}+r^{*}te^{\beta t}$.
  \end{thm}
  \begin{proof}
  	It is sufficient to show that the MLE of $\beta$ exists and unique under the assumptions in $(i)$ and $(ii)$. The derivative of $G(\beta)$ is obtained as 
  	\begin{align}\nonumber
  	\begin{aligned}
  	G'(\beta)
  	&=-\frac{N}{\beta}\bigg[1+\frac{\sum_{j=1}^{N}(1+r_{j})t^2_{j}e^{t_{j}}+p^2r^*e^{p}}{\sum_{j=1}^{N}(1+r_{j})(e^{t_{j}}-1)+r^*(e^{p}-1)}-\frac{\big[\sum_{j=1}^{N}(1+r_{j})t_{j}e^{t_{j}}+pr^*e^{p}\big]^2}{\big[\sum_{j=1}^{N}(1+r_{j})(e^{t_{j}}-1)+r^*(e^{p}-1)\big]^2}\bigg]\\
  	&=-\frac{N}{\beta}\bigg[\frac{G_{1}(\beta)+G_{2}(\beta)+G_{3}(\beta)+G_{4}(\beta)+G_{5}(\beta)}{\big[\sum_{j=1}^{N}(1+r_{j})(e^{t_{j}}-1)+r^*(e^{p}-1)\big]^2}\bigg],
  	\end{aligned}
  	\end{align}
  	where $t_{j}=\beta x_{j}$, $p=\beta t$ and
  	\begin{eqnarray*}
  	G_{1}(\beta)&=&\sum_{j=1}^{N}(1+r_{j})(e^{t_{j}}-1)\sum_{j=1}^{N}(1+r_{j})\frac{t^2_{j}e^{2t_{j}}}{(e^{t_{j}}-1)}-\bigg(\sum_{j=1}^{N}(1+r_{j})t_{j}e^{t_{j}}\bigg)^2,\\
  	G_{2}(\beta)&=&\sum_{j=1}^{N}(1+r_{j})(e^{t_{j}}-1)\bigg[\sum_{j=1}^{N}(1+r_{j})\frac{1}{(e^{t_{j}}-1)}\{(e^{t_{j}}-1)^2-t^2_{j}e^{t_{j}}\}\bigg],\\
  	G_{3}(\beta)&=&r^*\{(e^p-1)^2-p^2e^p\},\\
  	G_{4}(\beta)&=&r^*(e^p-1)\bigg[\sum_{j=1}^{N}(1+r_{j})t^2_{j}e^{t_{j}}+2\sum_{j=1}^{N}(1+r_{j})(e^{t_{j}}-1)\bigg]\\
  	G_{5}(\beta)&=&{\beta}^2r^*e^p\bigg[\sum_{j=1}^{N}(1+r_{j})t(e^{\beta x_{j}}-1)-2\sum_{j=1}^{N}(1+r_{j})x_{j}e^{\beta x_{j}}\bigg].
  	\end{eqnarray*}
  	Here $'$ stands for the derivative. Denote $a_{j}=\sqrt{(1+r_{j})(e^{t_{j}}-1)}$ and $b_{j}=\sqrt{(1+r_{j})\frac{t^2_{j}e^{2t_{j}}}{(e^{t_{j}}-1)}}$, for $1\leq j\leq N$. Then, by using Cauchy-Schwarz inequality for real numbers $a_{j}$ and $b_{j}$, it can be established that $G_{1}(\beta)\geq0$. Since $(e^x-1)^2>x^2e^x$ for all $x>0$, we have $G_{2}(\beta)>0$ and $G_{3}(\beta)>0$.  Clearly, $G_{4}(\beta)>0$. Further, $G_{5}(\beta)>0$ if and only if condition $(i)$ holds. Thus, $G'(\beta)<0$. This implies that $G(\beta)$ is decreasing in $\beta$ for all values of $x_{j}>0$ and $r_{j}>0$, $1\leq j\leq N$. Now, $G(\beta)$ has a unique root if and only if $G(0)>0$ and $G(\infty)<0$.\\
  	Using Lemma $3.2$ with $s_{j}=1+r_{j}$, we have
  	\begin{align}\nonumber
  	\begin{aligned}
  	G(0)&=\sum_{j=1}^{N}x_{j}+N\lim_{\beta\rightarrow 0}\biggl\{\frac{1}{\beta}-\frac{\sum_{j=1}^{N} (1+r_{j})x_{j}e^{\beta x_{j}}+tr^{*}e^{\beta t}}{\sum_{j=1}^{N}(1+r_{j})(e^{\beta x_{j}}-1)+t^*}\biggr\}\\
  	&=\sum_{j=1}^{N}x_{j}-N\frac{\sum_{j=1}^{N}(1+r_{j})x^{2}_{j}+r^*t^{2}}{2\big(\sum_{j=1}^{N}(1+r_{j})x_{j}+r^*t\big)}\\
  	&=\frac{2\sum_{j=1}^{N}x_{j}\big(\sum_{j=1}^{N}(1+r_{j})x_{j}+r^*t\big)-N\big(\sum_{j=1}^{N}(1+r_{j})x^2_{j}+r^*t^2\big)}{2\big(\sum_{j=1}^{N}(1+r_{j})x_{j}+r^*t\big)},
  	\end{aligned}
  	\end{align}
  	which is strictly positive under the assumption $(ii)$. Further, from the second part of Lemma $3.2$, we obtain
  	\begin{align}\nonumber
  	\begin{aligned}
  	G(\infty)&=\sum_{j=1}^{N}x_{j}-N\lim_{\beta\rightarrow \infty} \frac{\sum_{j=1}^{N}(1+r_{j})x_{j}e^{\beta x_{j}}+r^*te^{\beta t}}{\sum_{j=1}^{N}(1+r_{j})(e^{\beta x_{j}}-1)+t^*}\\
  	&=\left\{
  	\begin{array}{ll}
  		\sum_{j=1}^{N}(x_{j}-x_{N}),  ~~\mbox{for case I}\\
  		\sum_{j=1}^{N}(x_{j}-t),~~~~~\mbox{for case II and case III}
  	\end{array}
  	\right.\\
   	&<0.
  	\end{aligned}
  	\end{align}
  	Hence, the proof is completed.
  \end{proof}

  \subsection{Interval estimation}
  In this subsection, we obtain asymptotic confidence intervals of the model parameters. Some related results in this direction may be found in Dutta and Kayal \cite{dutta2023inference}, Elshahhat and Nassar \cite{elshahhat2024inference} and Alqasem and Elshahhat \cite{alqasem2025reliability}. Similar to these works, here we utilize asymptotic normality property of the MLEs of ${\alpha}_{1}$, ${\alpha}_{2}$ and $\beta$, to construct $100(1-\gamma) \%$ ACIs. Under some regularity conditions, the asymptotic distribution of the MLEs $(\widehat{{\alpha}_{1}},\widehat{{\alpha}_{2}},\widehat{\beta})^{T}$ can be obtained as
  \begin{align}
  	\nonumber (\widehat{{\alpha}_{1}},\widehat{{\alpha}_{2}}.\widehat{\beta})^{T}-({\alpha}_{1},{\alpha}_{2},{\beta})^{T} \sim N\big(0,I^{-1}_{obs}(\widehat{{\alpha}_{1}},\widehat{{\alpha}_{2}},\widehat{\beta})\big),
  \end{align}
  where
  \begin{align}\label{3.10}
  	I^{-1}_{obs}(\widehat{{\alpha}_{1}},\widehat{{\alpha}_{2}},\widehat{\beta})&= {\begin{bmatrix}
  			Q_{200} & Q_{110} & Q_{101}\\
  			Q_{110} & Q_{020} & Q_{011}\\
  			Q_{101} & Q_{011} & Q_{002}\\
  	\end{bmatrix}}^{-1}_{({\alpha}_{1},{\alpha}_{2},\beta)=(\widehat{{\alpha}_{1}},\widehat{{\alpha}_{2}},\widehat{\beta})}\nonumber\\&= {\begin{bmatrix}
  			Var(\widehat{{\alpha}_{1}}) & Cov(\widehat{{\alpha}_{1}},\widehat{{\alpha}_{2}}) & Cov(\widehat{{\alpha}_{1}},\widehat{\beta})\\
  			Cov(\widehat{{\alpha}_{1}},\widehat{{\alpha}_{2}}) & Var(\widehat{{\alpha}_{2}}) & Cov(\widehat{{\alpha}_{2}},\widehat{\beta})\\
  			Cov(\widehat{{\alpha}_{1}},\widehat{\beta}) & Cov(\widehat{{\alpha}_{2}},\widehat{\beta}) & Var(\widehat{\beta})\\
  	\end{bmatrix}},
  \end{align}
  is the inverse of observed Fisher information matrix for ${\alpha}_{1}$, ${\alpha}_{2}$ and $\beta$ and $Q_{ijk}=-\frac{\partial^{2} l}{\partial{\alpha}_{1}^{i} \partial{\alpha}_{2}^{j}\partial\beta^{k}}$, $i,j,k=0,1,2$; $T$ denotes tranpose of a vector.
  
  By using \eqref{3.3}, \eqref{3.4} and \eqref{3.5}, we can derive the following equations to determine the observed Fisher information matrix\\
  \begin{align}
  	\nonumber &Q_{200}=-\frac{\partial^{2} l}{\partial \alpha_{1}^{2}}=\frac{N_{1}}{\alpha_{1}^{2}}, Q_{110}=-\frac{\partial^{2} l}{\partial \alpha_{1} \partial \alpha_{2}}=0, Q_{020}=-\frac{\partial^{2} l}{\partial \alpha_{2}^{2}}=\frac{N_{2}}{\alpha_{2}^2},\\
  	\nonumber &Q_{101}=Q_{011}=-\frac{\partial^{2} l}{\partial \alpha_{1} \partial \beta}=-\frac{\partial^{2} l}{\partial \alpha_{2} \partial \beta}=\frac{A(x,\beta)}{\beta^{2}}-\frac{B(x,\beta)}{\beta}
  \end{align}
  and
  \begin{align}
  	\nonumber Q_{002}=-\frac{\partial^{2} l}{\partial \beta^{2}}=\frac{2\sum_{i=1}^{2}\alpha_{i} A(x,\beta)}{\beta^{3}}-\frac{2\sum_{i=1}^{2} B(x,\beta)}{\beta^{2}}+\frac{C(x,\beta)}{\beta},
  \end{align}
  where $C(x,\beta)=\sum_{j=1}^{N} (1+r_{j})x^{2}_{j}e^{\beta x_{j}}+r^{*}t^{2}e^{\beta t}$.
  Thus, the $100(1-\gamma) \%$ ACIs for ${\alpha}_{1}$, ${\alpha}_{2}$ and $\beta$ are respectively given by
  \begin{align}
  	\nonumber \bigg(\widehat{{\alpha}_{1}}~ \underline{+}~z_{\frac{\gamma}{2}}\sqrt{Var(\widehat{{\alpha}_{1}})}\bigg) ~~\mbox{,}~~ \bigg(\widehat{{\alpha}_{2}}~ \underline{+}~z_{\frac{\gamma}{2}}\sqrt{Var(\widehat{{\alpha}_{2}})}\bigg)~~\mbox{and}~~\bigg(\widehat{\beta}~ \underline{+}~z_{\frac{\gamma}{2}}\sqrt{Var(\widehat{\beta})}\bigg),
  \end{align}
  where $z_{\frac{\gamma}{2}}$ is the upper $\frac{\gamma}{2}$th percentile point of  standard normal distribution.\\
  \section{Bayesian estimation}
  In this section, we derive Bayes estimates for the unknown parameters $\alpha_{1}$, $\alpha_{2}$ and $\beta$ using IAT-II progressively censored competing risks data. Three different types of loss functions are employed in this process. The squared error loss function (SELF) is a symmetric or balance loss function, which provides equal importance to the over as well as under estimation. This loss function has been considered by various authors for the purpose of Bayesian analysis. For instance see Maiti and Kayal \cite{maiti2021estimation} and Saha and Yadav \cite{saha2021estimation}. The LINEX loss function (LLF), introduced by Varian \cite{varian1975bayesian} is an asymmetric loss function. It increases exponentially on one side of the origin and linearly (almost) on the other side of the origin. Due to this property, it becomes more appropriate tool than the SELF when there are consequences of differing overestimation and underestimation. Another important asymmetric loss function is the generalised entropy loss function (GELF) introduced by Calabria and Pulcini \cite{calabria1996point}. Like the SELF, the LLF and GELF have been used by several researchers. The interested readers are refered to Al Duais \cite{al2021bayesian}, Mohammed et. al \cite{mohammed2022bayesian}, Eliwa et. al \cite{eliwa2022general} and Sana and Faizan \cite{sana2021bayesian} for some works with respect to these loss functions. 
  
  Let $\widehat{\alpha}$ be an estimator of the parameter $\alpha$.  Then, the SELF, LLF and GELF are respectively defined as
  \begin{equation}
  	L_{SE}(\alpha,\widehat{\alpha}) =(\widehat{\alpha}-\alpha)^2, \label{4.1} 
  \end{equation}
  \begin{equation}
  	L_{LI}(\alpha,\widehat{\alpha})=e^{p(\widehat{\alpha}-\alpha)}-p(\widehat{\alpha}-\alpha)-1,~ p\neq 0, \label{4.2}
  \end{equation}
  \begin{equation}
  	L_{GE}(\alpha,\widehat{\alpha})=\bigg(\frac{\widehat{\alpha}}{\alpha}\bigg)^q -q\log \bigg(\frac{\widehat{\alpha}}{\alpha}\bigg)-1,~ q \neq 0. \label{4.3}
  \end{equation}
  We note that $p$ and $q$ respectively in (\ref{4.2}) and (\ref{4.3}) are known as the loss parameters. Depending on the signs of the loss parameters, the appropriate loss functions are generally chosen. Recall that the direction of symmetry is reflected by the sign of the parameter $p$ and the degree of asymmetry is reflected by its magnitude. For the case of the GELF, the departure from symmetry is reflected by the loss parameter $q$.  Under the above loss functions in \eqref{4.1}, \eqref{4.2} and \eqref{4.3}, we respectively write the Bayes estimates of $\alpha$ in the following forms
  \begin{equation}
  	\widehat{\alpha}_{SE}= E_{\alpha}(\alpha |\underline{x}), \label{4.4}
  \end{equation}
  \begin{equation}
  	\widehat{\alpha}_{LI}= -p^{-1} \log[E_{\alpha}(e^{{-p\alpha}}|\underline{x})],~p\neq 0, \label{4.5} 
  \end{equation}
  \begin{equation}
  	\widehat{\alpha}_{GE}= [E_{\alpha}({\alpha}^{-q}|\underline{x})]^{-\frac{1}{q }},~  q \neq 0,    \label{4.6}
  \end{equation}
  where the observed data is $\underline{x}=(x_{1},\cdots,x_{N})$. An important and challenging problem is to choose priors for the unknown model parameters for Bayesian estimation. It can be remarked that the prior information about the unknown model parameter of interest is reflected by the chosen prior distributions. For the theoretical purpose, here we consider informative independent gamma priors as $\alpha_{1} \sim G(a_{1},b_{1})$, $\alpha_{2} \sim G(a_{2},b_{2})$ and $\beta \sim G(a_{3},b_{3})$, where $a_{i},b_{i}>0$, $i=1,2,3$ are the hyper parameters, reflecting the prior knowledge about the model parameters. The $G(a_{i},b_{i})$ notation represents gamma distribution with scale parameter $1/b_{i}$ and shape parameter $a_{i}$.The gamma distribution's adaptability makes it ideal for this use. Combining its scale and shape characteristics with the model parameters allows for flexibility in capturing different previous assumptions. It is also flexible due to the presence of two parameters, which allow it to model various degrees of skewness and peakedness by changing the values of the shape and scale parameters. Furthermore, independent gamma priors are special examples of non-informative priors on the scale and shape parameters. Because of this, the gamma distribution is a reliable option even in situations where little to no prior information is available. Note that many authors have considered independent gamma priors for the purpose of Bayesian estimation. See, for example Kundu and Gupta \cite{kundu2008generalized}, Saha and Dutta \cite{saha2025bayesian} and Pradhan and Kundu \cite{pradhan2011bayes}. The joint prior PDF of ${\alpha}_{1}$, ${\alpha}_{2}$ and $\beta$ is written as
  \begin{equation}
  	\pi(\alpha_{1},\alpha_{2},\beta) \propto \alpha_{1}^{a_{1}-1}\alpha_{2}^{a_{2}-1}\beta^{a_{3}-1}e^{-(b_{1}\alpha_{1}+b_{2}\alpha_{2}+b_{3}\beta)}, ~~\alpha_{1},\alpha_{2},\beta,a_{i},b_{i}>0, ~i=1,2,3. \label{4.7}
  \end{equation}
  Using likelihood function from \eqref{3.1} and joint prior density function from \eqref{4.7}, we obtain posterior probability density function as
  \begin{equation}
  	\pi({\alpha}_{1},{\alpha}_{2},\beta|\underline{x})=  K^{-1}\alpha_{1}^{a_{1}+N_{1}-1}\alpha_{2}^{a_{2}+N_{2}-1}\beta^{a_{3}-1} e^{-\beta(b_{3}-\sum_{j=1}^{N}x_{j})}e^{-(b_{1}\alpha_{1}+b_{2}\alpha_{2})} e^{-\frac{A(x,\beta)(\alpha_{1}+\alpha_{2})}{\beta}} \label{4.8}
  \end{equation}
  where
  \begin{align}
  	\nonumber K=\int_{0}^{\infty}\int_{0}^{\infty}\int_{0}^{\infty} \alpha_{1}^{a_{1}+N_{1}-1}\alpha_{2}^{a_{2}+N_{2}-1}\beta^{a_{3}-1}
  	e^{-\beta(b_{3}-\sum_{i=1}^{N}x_{i})}e^{-(b_{1}\alpha_{1}+b_{2}\alpha_{2})} e^{-\frac{A(x,\beta)(\alpha_{1}+\alpha_{2})}{\beta}}~d\alpha_{1} d\alpha_{2} d\beta
  \end{align}
  and $A(x,\beta)= \sum_{j=1}^{N} (r_{j}+1)(e^{\beta x_{j}}-1) +t^{*}$.
  
  Since, we can not derive the solutions for Bayesian estimators explicitly, a numerical method is applied to derive the Bayes estimates under SELF, LLF and GELF. Here, we try to incorporate Metropolis-Hastings(M-H algorithm) into the Gibbs method to produce the Bayesian estimators for $\alpha_{1}$, $\alpha_{2}$ and $\beta$. We write the conditional posterior distributions of $\alpha_{1}$, $\alpha_{2}$ and $\beta$ as
  \begin{equation}
  	\pi_{1}(\alpha_{1}|\alpha_{2},\beta,\underline{x})\propto \alpha_{1}^{a_{1}+N_{1}-1}\exp \bigg[-\alpha_{1}\bigg(b_{1}+\frac{A(x,\beta)}{\beta}\bigg)\bigg] \sim G\bigg(N_{1}+a_{1},\bigg(b_{1}+\frac{A(x,\beta)}{\beta}\bigg)\bigg), \label{4.9}
  \end{equation}
  \vspace{0.2cm}
  \begin{equation}
  	\pi_{2}(\alpha_{2}|\alpha_{1},\beta,\underline{x})\propto \alpha_{2}^{a_{2}+N_{2}-1}\exp \bigg[-\alpha_{2}\bigg(b_{2}+\frac{A(x,\beta)}{\beta}\bigg)\bigg] \sim G\bigg(N_{2}+a_{2},\bigg(b_{2}+\frac{A(x,\beta)}{\beta}\bigg)\bigg) \label{4.10}
  \end{equation}
  and
  \begin{equation}
  	\pi_{3}(\beta|\alpha_{1},\alpha_{2},\underline{x})\propto \beta^{a_{3}-1}e^{-\beta(b_{3}-\sum_{j=1}^{N}x_{j})}e^{-\frac{A(x,\beta)(\alpha_{1}+\alpha_{2})}{\beta}}. \label{4.11}
  \end{equation}
  We notice from \eqref{4.9} and \eqref{4.10} that the samples for $\alpha_{1}$ and $\alpha_{2}$ can be easily generated by using gamma distributions. However, $\pi_{3}(\beta|\alpha_{1},\alpha_{2},\underline{x})$ in \eqref{4.11} can not be converted to well-known distribution. Thus, it is not possible to generate sample for $\beta$ directly by standard methods. One way to look at the density plot of the conditional posterior density of $\beta$, which is depicted in Figure \ref{fig:fig3}. It looks like to the Gaussian distribution. As a result, the normal proposal distribution can be used in M-H algorithm for the purpose of Bayesian estimation.
  \begin{figure}[ht!]
   \begin{center}
   \includegraphics[height=3.9in,width=4.9in]{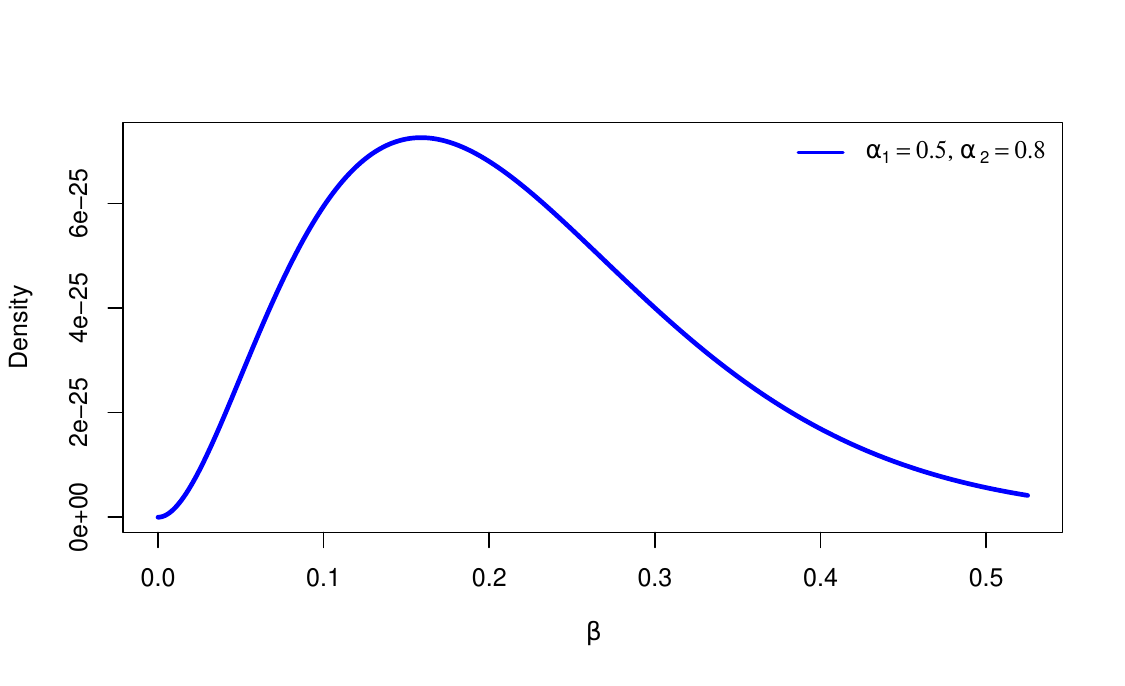}
   \caption{Plot for the conditional posterior probability density function of $\beta$ for $n=60$, $m=30$ and CS-II}
   \label{fig:fig3}
   \end{center}
  \end{figure}
  
  We take an arbitrary function $\phi(\alpha_{1},\alpha_{2},\beta)$ of parameters $\alpha_{1}$, $\alpha_{2}$ and $\beta$. Then, from \eqref{4.4} to \eqref{4.6}, the Bayes estimates of $\phi(\alpha_{1},\alpha_{2},\beta)$ regarding SELF, LLF and GELF are given by 
  \begin{equation}
  	\widehat{\phi}_{SE}= \int_{0}^{\infty}\int_{0}^{\infty}\int_{0}^{\infty} \phi(\alpha_{1},\alpha_{2},\beta) \pi({\alpha}_{1},{\alpha}_{2},\beta|\underline{x}) d\alpha_{1}d\alpha_{2}d\beta, \label{4.12}
  \end{equation}
  \begin{equation}
  	\widehat{\phi}_{LI}= -\bigg(\frac{1}{p}\bigg)\log\bigg[\int_{0}^{\infty}\int_{0}^{\infty}\int_{0}^{\infty} e^{-p\phi(\alpha_{1},\alpha_{2},\beta)} \pi({\alpha}_{1},{\alpha}_{2},\beta|\underline{x}) d\alpha_{1}d\alpha_{2}d\beta\bigg] \label{4.13}
  \end{equation}
  and
  \begin{equation}
  	\widehat{\phi}_{GE}= \bigg[\int_{0}^{\infty}\int_{0}^{\infty}\int_{0}^{\infty} (\phi(\alpha_{1},\alpha_{2},\beta))^{-q} \pi({\alpha}_{1},{\alpha}_{2},\beta|\underline{x}) d\alpha_{1}d\alpha_{2}d\beta\bigg]^{-\frac{1}{q}} \label{4.14}
  \end{equation}
  respectively.
  
  The MCMC algorithm for Bayesian estimates and Bayesian credible intervals will execute the following steps:\\
  ------------------------------------------------------------------------------------------------------------------------
  \textbf{Algorithm II}\\
  ------------------------------------------------------------------------------------------------------------------------
  \textbf{Step-1:} Select initial guesses of $(\alpha_{1},\alpha_{2},\beta)$, indicated by $(\alpha_{1}^{(0)},\alpha_{2}^{(0)},\beta^{(0)})$ and set $k=1$.\\
  \textbf{Step-2:} Generate $\beta^{(k)}$ from $\pi_{3}(\beta^{(k-1)}|\alpha_{1}^{(k-1)},\alpha_{2}^{(k-1)},\underline{x})$ using M-H method using normal proposal distribution $N(\beta^{(k-1)},Var(\beta))$, where $\beta^{(k-1)}$ is the current value of $\beta$ and $Var(\beta)$ is the variance of $\beta$.\\
  \textbf{Step-3:} Generate $\alpha_{i}^{(k)}$ from $G\bigg(N_{i}+a_{i},\bigg(b_{i}+\frac{A(x,\beta)}{\beta}\bigg)\bigg)$ for $i=1,2$.\\
  \textbf{Step-4:} Set $k=k+1$.\\
  \textbf{Step-5:} Repeat Steps $(2-3)$ $D$ times to get the necessary number of samples $$(\alpha_{1}^{(1)},\alpha_{2}^{(1)},\beta^{(1)}),(\alpha_{1}^{(2)},\alpha_{2}^{(2)},\beta^{(2)}),\cdots,(\alpha_{1}^{(D)},\alpha_{2}^{(D)},\beta^{(N)}).$$ 
  The Bayes estimates are generated using the remaining $D-M$ burn-in samples once the initial $M$ burn-in samples have been discarded.\\
  \textbf{Step-6:} The Bayes estimates of any function $\phi(\alpha_{1},\alpha_{2},\beta)$ for the SELF, LLF and GELF are calculated using 
  \begin{align}
  	\nonumber &\widehat{\phi}_{SE}=\frac{1}{D-M}\sum_{k=M+1}^{D} \phi(\alpha_{1}^{(k)},\alpha_{2}^{(k)},\beta^{(k)}),\\
  	\nonumber &\widehat{\phi}_{LI}=\frac{-1}{c}\log\bigg[\frac{1}{D-M}\sum_{k=M+1}^{D}e^{-c\phi(\alpha_{1}^{(k)},\alpha_{2}^{(k)},\beta^{(k)})}\bigg]
  \end{align}
  and
  \begin{align}\nonumber
  	\widehat{\phi}_{GE}=\bigg[\frac{1}{D-M}\sum_{k=M+1}^{D} (\phi(\alpha_{1}^{(k)},\alpha_{2}^{(k)},\beta^{(k)}))^{-c}\bigg]^{-\frac{1}{c}}
  \end{align}
  where $\phi(\alpha_{1},\alpha_{2},\beta)$ refers to the parameters $\alpha_{1},\alpha_{2},\beta$.\\
  \textbf{Step-7:} Order $\phi^{M+1},\phi^{M+2},\cdots,\phi^{D}$ as $\phi_{(1)},\phi_{(2)},\cdots,\phi_{(D-M)}$. Then, the $100(1-\gamma)\%$ Bayesian credible interval of $\widehat{\phi}$ is given by $(\phi_{[(D-M)\frac{\gamma}{2}]},\phi_{[(D-M)(1-\frac{\gamma}{2})]})$.\\
  ------------------------------------------------------------------------------------------------------------------------
  
  \section{Hypothesis testing}
  In previous sections, inference for Gompertz competing risks data based on IAT-II PCS has been obtained with common scale parameter $\beta_{1}=\beta_{2}=\beta$ and different shape parameters $\alpha_{1}$ and $\alpha_{2}$. It is important to examine whether the scale parameters $\beta_{1}$ and $\beta_{2}$ are equal in practice, and likewise to test if the shape parameters $\alpha_{1}$ and $\alpha_{2}$ are equal in practice. To investigate this phenomenon, following hypotheses have been proposed as
  \begin{align}
  	\nonumber&(1)~~H_{0}: \alpha_{1}=\alpha_{2}=\alpha~~\mbox{vs}~~H_{1}: \alpha_{1}\neq\alpha_{2}\\
  	&(2)~~H_{0}: \beta_{1}=\beta_{2}=\beta~~\mbox{vs}~~H_{1}: \beta_{1}\neq\beta_{2}\label{5.1}
  \end{align}
  For test $(1)$, the likelihood ratio statistics (LRS) is given  by
  \begin{equation}
  	\Lambda_{\alpha}=\frac{max~L(\alpha,\beta_{1},\beta_{2})}{max~L(\alpha_{1},\alpha_{2},\beta_{1},\beta_{2})},\label{5.2}
  \end{equation}
  where $L(\cdot)$ denotes the likelihood function. For large $n$, the LRS can be expressed as
  \begin{equation}
  	LRS_{\alpha}=-2\log\Lambda_{\alpha}=-2(l(\hat{\alpha},\hat{\beta_{1}},\hat{\beta_{2}})-l(\hat{\alpha_{1}},\hat{\alpha_{2}},\hat{\beta_{1}},\hat{\beta_{2}}))\sim \chi^{2}_{(1)},\label{5.3}
  \end{equation} 
  where $l(\cdot)$ denotes the log-likelihood function. Similarly, for test $(2)$,
  \begin{equation}
  	\Lambda_{\beta}=\frac{max~L(\alpha_{1},\alpha_{2},\beta)}{max~L(\alpha_{1},\alpha_{2},\beta_{1},\beta_{2})}\label{5.4}
  \end{equation}
  and
  \begin{equation}
  	LRS_{\beta}=-2\log\Lambda_{\beta}=-2(l(\hat{\alpha_{1}},\hat{\alpha_{2}},\hat{\beta})-l(\hat{\alpha_{1}},\hat{\alpha_{2}},\hat{\beta_{1}},\hat{\beta_{2}}))\sim\chi^{2}_{(1)}.\label{5.5}
  \end{equation}
  Using the asymptotic distribution of LRS, the likelihood ratio test can be constructed and the null hypothesis $H_{0}$ can be rejected if $LRS>c$, where $P(\chi^{2}_{(1)}>c)=$size of the test at the $5\%$ significance level.
  \section{Simulation study}
  Here, we have conducted a Monte Carlo simulation to analyze the performance of the proposed estimates based on 5,000 generated IAT-II PCS samples using $R$ programming. The performance of the  proposed point estimates are evaluated using  mean squared error (MSE). The performance of the proposed interval estimates are evaluated based on coverage probability (CP) and average length (AL).
  \begin{itemize}
  	\item \textbf{MSE:} Let $\widehat{{\alpha}_{i}}$ be an estimate of the parameter, say $\alpha_{i}$. Then the MSE is given by $\frac{1}{M}\sum_{i=1}^{M}(\widehat{{\alpha}_{i}}-\alpha_{i})^2$, where $M$ indicates the number of repetitions in the simulation. A more accurate estimate is detected by smaller MSE.
  	\item \textbf{CP:} It represents the likelihood that the true parameter value lies within the interval estimates. When the CP is close to the nominal value, that is, $100(1-\gamma)\%$, it indicates a more accurate result in terms of the CP.
  	\item \textbf{AL:} A shorter interval length suggests that the prediction model demonstrates greater accuracy when aligned with the experimental data.
  \end{itemize}
  To generate competing risks samples from the Gompertz distribution under IAT-II PCS, the following algorithm is introduced.\\
  ------------------------------------------------------------------------------------------------------------------------
  \textbf{Algorithm III}\\
  ------------------------------------------------------------------------------------------------------------------------
  \textbf{Step 1:} Generate type-II progressive censored Gompertz data of effective sample size $m$ with parameter $\alpha_{1}+\alpha_{2}$. \\
  \textbf{Step 2:} Get IAT-II PCS sample of size $N$ by using two time thresholds $t_1$ and $t_2$ on $X_{m:m:n}$.\\
  \textbf{Step 3:} Allot failure risk $\delta_{j}=0,1,(j=1,\cdots,m)$, with probability of failure risks as $P(\delta_{j}=1)=\frac{\alpha_{1}}{\alpha_{1}+\alpha_{2}}$ and $P(\delta_{j}=0)=\frac{\alpha_{2}}{\alpha_{1}+\alpha_{2}}$ to the generated censored data.\\
  ------------------------------------------------------------------------------------------------------------------------
  For simulation, different sample sizes $(n)$, different effective sample sizes $(m)$ and different time thresholds $(t_{1},t_{2}),$ have been considered. Three different censoring schemes (CS) are adopted, given as follows\\[1 mm]
  \textbf{Scheme-I}~~~$r_{1}=r_{2}=\cdots=r_{m-1}=0,r_{m}=n-m$.\\
  \textbf{Scheme-II}~~$r_{1}=r_{2}=\cdots=r_{m}=(n-m)/m$.\\
  \textbf{Scheme-III}~$r_{1}=r_{2}=\cdots=r_{m-1}=1,r_{m}=n-2m+1$.\\
  
  The average values and MSEs ofthe MLEs and Bayes estimates are computed when the true values of $(\alpha_{1},\alpha_{2},\beta)$ are $(0.5,0.8,0.75)$ and $(0.7,1.15,0.9)$. For simulation, Bayes estimates are computed based on both non-informative priors (NIP) and informative priors (IP). For NIP, the hyper parameters are set to (nearly) zero, while for IP, they are calculated using the following formula (see Dey et al. \cite{dey2016estimation}).\\
  \begin{align}
  	\nonumber &a=\frac{\widehat{\alpha_{1}}^2}{\frac{1}{N-1}\sum_{j=1}^{N}(\widehat{\alpha}_{1j}-\widehat{\alpha_{1}})^2},~~b=\frac{\widehat{\alpha_{1}}}{\frac{1}{N-1}\sum_{j=1}^{N}(\widehat{\alpha}_{1j}-\widehat{\alpha_{1}})^2},\\
  	\nonumber &c=\frac{\widehat{\alpha_{2}}^2}{\frac{1}{N-1}\sum_{j=1}^{N}(\widehat{\alpha}_{2j}-\widehat{\alpha_{2}})^2},~~d=\frac{\widehat{\alpha_{2}}}{\frac{1}{N-1}\sum_{j=1}^{N}(\widehat{\alpha}_{2j}-\widehat{\alpha_{2}})^2},\\
  	\nonumber &e=\frac{\widehat{\beta}^2}{\frac{1}{N-1}\sum_{j=1}^{N}(\widehat{\beta}_{j}-\widehat{\beta})^2},~~f=\frac{\widehat{\beta}}{\frac{1}{N-1}\sum_{j=1}^{N}(\widehat{\beta}_{j}-\widehat{\beta})^2},
  \end{align}
  where $\alpha_{1j}$, $\alpha_{2j}$ and $\beta_{j}$, $j=1,\cdots,N$ are the MLEs of $\alpha_{1}$, $\alpha_{2}$ and $\beta$, respectively for $N$ number of IAT-II PCSs. We compute Bayes estimates using SELF, LLF (when $p=-0.05,0.25$) and GELF (when $q=-0.05,0.25$). Next, the ACIs and HPD credible intervals are computed based on the AL and CP. The following observations are made from Tables \ref{T1} to \ref{T4}.
  \begin{itemize}
   \item For a particular value of $n$, when effective sample size $m$ is increased, the simulated values of MSEs decrease.
   \item The Bayes estimates based on IP perform better than that based on NIP.
   \item Based on the MSEs and average values of parameters, the Bayes estimates under LLF provide better results than other estimates of $\alpha_{1}$, $\alpha_{2}$ and $\beta$.
  \end{itemize}
  For interval estimates, the following observations are made from Tables \ref{T5} to \ref{T8}.
  \begin{itemize}
   \item When sample size $(n)$ and effective sample sample size $(m)$ increase, ALs of intervals decrease.
   \item The HPD credible intervals perform better than the ACIs in terms of the AL and CP.
   \item When the values of time thresholds $t_{1}$ and $t_{2}$  increase, the length of intervals decrease.
  \end{itemize}

  The performance of the classical and Bayesian estimates are quite reliable. When $p$ approaches zero, the Bayes estimates under the LLF are nearly identical to those under the SELF. The Bayes estimates based on IP with respect to LLF is superior choice in this study. Besides, the HPD credible intervals are recommended to prefer if one considers AL and CP.
   {
   \scriptsize
\begin{longtable}{*{14}c*{10}{r@{}l}}
   	\caption{Average values and MSEs of the MLEs and Bayes estimates for different schemes with different values of $n$, $m$ when $t_{1}=0.8$, $t_{2}=1.12$ and $\Theta=(\alpha_{1},\alpha_{2},\beta)=(0.5, 0.8, 0.75)$.}
   	\label{T1} \\
   	\toprule
   	\multicolumn{4}{c}{} & \multicolumn{2}{c}{SELF} & \multicolumn{2}{c}{LLF ($p=-0.05$)} & \multicolumn{2}{c}{LLF ($p=0.25$)} & \multicolumn{2}{c}{GELF ($q=-0.05$)} & \multicolumn{2}{c}{GELF ($q=0.25$)} & \\
   	\cmidrule(lr){5-6} \cmidrule(lr){7-8} \cmidrule(lr){9-10} \cmidrule(lr){11-12} \cmidrule(lr){13-14}
   	\multicolumn{1}{c}{$(n,m)$} & \multicolumn{1}{c}{CS} & \multicolumn{1}{c}{$\Theta$} & \multicolumn{1}{c}{MLE} & \multicolumn{1}{c}{NIP} & \multicolumn{1}{c}{IP} & \multicolumn{1}{c}{NIP} & \multicolumn{1}{c}{IP} & \multicolumn{1}{c}{NIP} & \multicolumn{1}{c}{IP} & \multicolumn{1}{c}{NIP} & \multicolumn{1}{c}{IP} & \multicolumn{1}{c}{NIP} & \multicolumn{1}{c}{IP} \\
   	\midrule
   	\endfirsthead
   	\caption[]{(Continued)} \\
   	\toprule
   	\multicolumn{4}{c}{} & \multicolumn{2}{c}{SELF} & \multicolumn{2}{c}{LLF ($p=-0.05$)} & \multicolumn{2}{c}{LLF ($p=0.25$)} & \multicolumn{2}{c}{GELF ($q=-0.05$)} & \multicolumn{2}{c}{GELF ($q=0.25$)} \\
   	\cmidrule(lr){5-6} \cmidrule(lr){7-8} \cmidrule(lr){9-10} \cmidrule(lr){11-12} \cmidrule(lr){13-14}
   	\multicolumn{1}{c}{$(n,m)$} & \multicolumn{1}{c}{CS} & \multicolumn{1}{c}{$\Theta$} & \multicolumn{1}{c}{MLE} & \multicolumn{1}{c}{NIP} & \multicolumn{1}{c}{IP} & \multicolumn{1}{c}{NIP} & \multicolumn{1}{c}{IP} & \multicolumn{1}{c}{NIP} & \multicolumn{1}{c}{IP} & \multicolumn{1}{c}{NIP} & \multicolumn{1}{c}{IP} & \multicolumn{1}{c}{NIP} & \multicolumn{1}{c}{IP} \\
   	\midrule
   	\endhead
   	\endfoot
   	\midrule
   	\endlastfoot
   	(30,10) & I & $\alpha_{1}$ & 0.4151& 0.4347& 0.4381& 0.4365& 0.4397& 0.4263& 0.4319& 0.3683& 0.3876& 0.3467& 0.3716 \\
   	& & & 0.0305& 0.0263& 0.0165& 0.0265& 0.0167& 0.0276& 0.0166& 0.0332& 0.0227& 0.0390& 0.0257 \\
   	& & $\alpha_{2}$ & 0.9687& 1.0144& 0.8965& 1.0185& 0.8984& 0.9847& 0.8872& 0.9465& 0.8581& 0.9254& 0.8458 \\
   	& & & 0.1555& 0.1663& 0.0380& 0.1703& 0.0390& 0.1592& 0.0351& 0.1263& 0.0300& 0.9271& 0.0276 \\
   	& & $\beta$ & 1.0760& 0.7309& 0.7388& 0.7314& 0.7382& 0.7255& 0.7372& 0.7372& 0.7272& 0.7289& 0.7247 \\
   	& & & 0.2558& 0.0372& 0.0275& 0.0372& 0.0269& 0.0370& 0.0275& 0.0308& 0.0285& 0.0314& 0.0298 \\
   	& II & $\alpha_{1}$ & 0.4177& 0.4354& 0.4405& 0.4389& 0.4424& 0.4268& 0.4342& 0.3702& 0.3901& 0.3470& 0.3730 \\
   	& & & 0.0308& 0.0284& 0.0163& 0.0277& 0.0158& 0.0275& 0.0164& 0.0337& 0.0217& 0.0388& 0.0253 \\
\end{longtable}
   }
   {
   	\scriptsize
   	\begin{longtable}{*{14}c*{10}{r@{}l}}
   	\caption{Average values and MSEs of the MLEs and Bayes estimates for different schemes with different values of $n$, $m$ when $t_{1}=0.8$, $t_{2}=1.12$ and $\Theta=(\alpha_{1},\alpha_{2},\beta)=(0.7, 1.15, 0.9)$.}
   	\label{T2} \\
   	\toprule
   	\multicolumn{4}{c}{} & \multicolumn{2}{c}{SELF} & \multicolumn{2}{c}{LLF ($p=-0.05$)} & \multicolumn{2}{c}{LLF ($p=0.25$)} & \multicolumn{2}{c}{GELF ($q=-0.05$)} & \multicolumn{2}{c}{GELF ($q=0.25$)} & \\
   	\cmidrule(lr){5-6} \cmidrule(lr){7-8} \cmidrule(lr){9-10} \cmidrule(lr){11-12} \cmidrule(lr){13-14}
   	\multicolumn{1}{c}{$(n,m)$} & \multicolumn{1}{c}{CS} & \multicolumn{1}{c}{$\Theta$} & \multicolumn{1}{c}{MLE} & \multicolumn{1}{c}{NIP} & \multicolumn{1}{c}{IP} & \multicolumn{1}{c}{NIP} & \multicolumn{1}{c}{IP} & \multicolumn{1}{c}{NIP} & \multicolumn{1}{c}{IP} & \multicolumn{1}{c}{NIP} & \multicolumn{1}{c}{IP} & \multicolumn{1}{c}{NIP} & \multicolumn{1}{c}{IP} \\
   	\midrule
   	\endfirsthead
   	\caption[]{(Continued)} \\
   	\toprule
   	\multicolumn{4}{c}{} & \multicolumn{2}{c}{SELF} & \multicolumn{2}{c}{LLF ($p=-0.05$)} & \multicolumn{2}{c}{LLF ($p=0.25$)} & \multicolumn{2}{c}{GELF ($q=-0.05$)} & \multicolumn{2}{c}{GELF ($q=0.25$)} \\
   	\cmidrule(lr){5-6} \cmidrule(lr){7-8} \cmidrule(lr){9-10} \cmidrule(lr){11-12} \cmidrule(lr){13-14}
   	\multicolumn{1}{c}{$(n,m)$} & \multicolumn{1}{c}{CS} & \multicolumn{1}{c}{$\Theta$} & \multicolumn{1}{c}{MLE} & \multicolumn{1}{c}{NIP} & \multicolumn{1}{c}{IP} & \multicolumn{1}{c}{NIP} & \multicolumn{1}{c}{IP} & \multicolumn{1}{c}{NIP} & \multicolumn{1}{c}{IP} & \multicolumn{1}{c}{NIP} & \multicolumn{1}{c}{IP} & \multicolumn{1}{c}{NIP} & \multicolumn{1}{c}{IP} \\
   	\midrule
   	\endhead
   	\endfoot
   	\midrule
   	\endlastfoot
   	(30,10) & I & $\alpha_{1}$ & 0.6074& 0.6132& 0.6563& 0.6140& 0.6574& 0.5964& 0.6518& 0.5174& 0.6254& 0.4891& 0.6165 \\
   	& & & 0.0634& 0.0525& 0.0073& 0.0508& 0.0072& 0.0506& 0.0070& 0.0637& 0.0105& 0.0731& 0.0112 \\
   	& & $\alpha_{2}$ & 1.4172& 1.4306& 1.1579& 1.4326& 1.1582& 1.3915& 1.1567& 1.3295& 1.1533& 1.3044& 1.1521 \\
   	& & & 0.3696& 0.3236& 0.0005& 0.3164& 0.0004& 0.2754& 0.0004& 0.2332& 0.0004& 0.2274& 0.0003 \\
   	& & $\beta$ & 1.1322& 0.8995& 0.9001& 0.8995& 0.9001& 0.8995& 0.8995& 0.8994& 0.9000& 0.8994& 0.8994 \\
   	& & & 0.1565& 0.0004& 0.0003& 0.0003& 0.0003& 0.0004& 0.0003& 0.0004& 0.0003& 0.0003& 0.0003 \\
   	& II & $\alpha_{1}$ & 0.6103& 0.6244& 0.6576& 0.6281& 0.6591& 0.6020& 0.6522& 0.5292& 0.6271& 0.4937& 0.6168 \\
   	& & & 0.0605& 0.0578& 0.0067& 0.0588& 0.0065& 0.0523& 0.0071& 0.0666& 0.0097& 0.0734& 0.0113 \\
\end{longtable}
   }
   {
   	\scriptsize
   	\begin{longtable}{*{14}c*{10}{r@{}l}}
   	\caption{Average values and MSEs of the MLEs and Bayes estimates for different schemes with different values of $n$, $m$ when $t_{1}=0.4$, $t_{2}=0.6$ and $\Theta=(\alpha_{1},\alpha_{2},\beta)=(0.5, 0.8, 0.75)$.}
   	\label{T3} \\
   	\toprule
   	\multicolumn{4}{c}{} & \multicolumn{2}{c}{SELF} & \multicolumn{2}{c}{LLF ($p=-0.05$)} & \multicolumn{2}{c}{LLF ($p=0.25$)} & \multicolumn{2}{c}{GELF ($q=-0.05$)} & \multicolumn{2}{c}{GELF ($q=0.25$)} & \\
   	\cmidrule(lr){5-6} \cmidrule(lr){7-8} \cmidrule(lr){9-10} \cmidrule(lr){11-12} \cmidrule(lr){13-14}
   	\multicolumn{1}{c}{$(n,m)$} & \multicolumn{1}{c}{CS} & \multicolumn{1}{c}{$\Theta$} & \multicolumn{1}{c}{MLE} & \multicolumn{1}{c}{NIP} & \multicolumn{1}{c}{IP} & \multicolumn{1}{c}{NIP} & \multicolumn{1}{c}{IP} & \multicolumn{1}{c}{NIP} & \multicolumn{1}{c}{IP} & \multicolumn{1}{c}{NIP} & \multicolumn{1}{c}{IP} & \multicolumn{1}{c}{NIP} & \multicolumn{1}{c}{IP} \\
   	\midrule
   	\endfirsthead
   	\caption[]{(Continued)} \\
   	\toprule
   	\multicolumn{4}{c}{} & \multicolumn{2}{c}{SELF} & \multicolumn{2}{c}{LLF ($p=-0.05$)} & \multicolumn{2}{c}{LLF ($p=0.25$)} & \multicolumn{2}{c}{GELF ($q=-0.05$)} & \multicolumn{2}{c}{GELF ($q=0.25$)} \\
   	\cmidrule(lr){5-6} \cmidrule(lr){7-8} \cmidrule(lr){9-10} \cmidrule(lr){11-12} \cmidrule(lr){13-14}
   	\multicolumn{1}{c}{$(n,m)$} & \multicolumn{1}{c}{CS} & \multicolumn{1}{c}{$\Theta$} & \multicolumn{1}{c}{MLE} & \multicolumn{1}{c}{NIP} & \multicolumn{1}{c}{IP} & \multicolumn{1}{c}{NIP} & \multicolumn{1}{c}{IP} & \multicolumn{1}{c}{NIP} & \multicolumn{1}{c}{IP} & \multicolumn{1}{c}{NIP} & \multicolumn{1}{c}{IP} & \multicolumn{1}{c}{NIP} & \multicolumn{1}{c}{IP} \\
   	\midrule
   	\endhead
   	\endfoot
   	\midrule
   	\endlastfoot
   	(30,10) & I & $\alpha_{1}$ & 0.4216& 0.4293& 0.4285& 0.4311& 0.4337& 0.4116& 0.4223& 0.3637& 0.3823& 0.3346& 0.3627 \\
   	& & & 0.0301& 0.0339& 0.0234& 0.0342& 0.0222& 0.0333& 0.0234& 0.0394& 0.0273& 0.0450& 0.0320 \\
   	& & $\alpha_{2}$ & 0.9789& 1.0002& 0.8731& 1.0044& 0.8819& 0.9579& 0.8640& 0.9330& 0.8421& 0.8903& 0.8236 \\
   	& & & 0.1672& 0.2003& 0.0559& 0.2047& 0.0554& 0.1671& 0.0529& 0.1576& 0.0463& 0.1364& 0.0458 \\
   	& & $\beta$ & 1.0517& 0.7242& 0.7377& 0.7247& 0.7333& 0.7238& 0.7359& 0.7307& 0.7219& 0.7274& 0.7238 \\
   	& & & 0.2493& 0.0366& 0.0279& 0.0365& 0.0273& 0.0363& 0.0279& 0.0301& 0.0292& 0.0308& 0.0302 \\
   	& II & $\alpha_{1}$ & 0.4243& 0.3287& 0.3454& 0.3300& 0.3444& 0.3181& 0.3488& 0.2761& 0.3218& 0.2559& 0.2938 \\
   	& & & 0.0696& 0.0681& 0.0494& 0.0683& 0.0496& 0.0703& 0.0496& 0.0787& 0.0593& 0.0864& 0.0628 \\
\end{longtable}
   }
   {
   	\scriptsize
   	\begin{longtable}{*{14}c*{10}{r@{}l}}
    \caption{Average values and MSEs of the MLEs and Bayes estimates for different schemes with different values of $n$, $m$ when $t_{1}=0.4$, $t_{2}=0.6$ and $\Theta=(\alpha_{1},\alpha_{2},\beta)=(0.7, 1.15, 0.9)$.}
    \label{T4} \\
    \toprule
    \multicolumn{4}{c}{} & \multicolumn{2}{c}{SELF} & \multicolumn{2}{c}{LLF ($p=-0.05$)} & \multicolumn{2}{c}{LLF ($p=0.25$)} & \multicolumn{2}{c}{GELF ($q=-0.05$)} & \multicolumn{2}{c}{GELF ($q=0.25$)} & \\
    \cmidrule(lr){5-6} \cmidrule(lr){7-8} \cmidrule(lr){9-10} \cmidrule(lr){11-12} \cmidrule(lr){13-14}
    \multicolumn{1}{c}{$(n,m)$} & \multicolumn{1}{c}{CS} & \multicolumn{1}{c}{$\Theta$} & \multicolumn{1}{c}{MLE} & \multicolumn{1}{c}{NIP} & \multicolumn{1}{c}{IP} & \multicolumn{1}{c}{NIP} & \multicolumn{1}{c}{IP} & \multicolumn{1}{c}{NIP} & \multicolumn{1}{c}{IP} & \multicolumn{1}{c}{NIP} & \multicolumn{1}{c}{IP} & \multicolumn{1}{c}{NIP} & \multicolumn{1}{c}{IP} \\
    \midrule
    \endfirsthead
    \caption[]{(Continued)} \\
    \toprule
    \multicolumn{4}{c}{} & \multicolumn{2}{c}{SELF} & \multicolumn{2}{c}{LLF ($p=-0.05$)} & \multicolumn{2}{c}{LLF ($p=0.25$)} & \multicolumn{2}{c}{GELF ($q=-0.05$)} & \multicolumn{2}{c}{GELF ($q=0.25$)} \\
    \cmidrule(lr){5-6} \cmidrule(lr){7-8} \cmidrule(lr){9-10} \cmidrule(lr){11-12} \cmidrule(lr){13-14}
    \multicolumn{1}{c}{$(n,m)$} & \multicolumn{1}{c}{CS} & \multicolumn{1}{c}{$\Theta$} & \multicolumn{1}{c}{MLE} & \multicolumn{1}{c}{NIP} & \multicolumn{1}{c}{IP} & \multicolumn{1}{c}{NIP} & \multicolumn{1}{c}{IP} & \multicolumn{1}{c}{NIP} & \multicolumn{1}{c}{IP} & \multicolumn{1}{c}{NIP} & \multicolumn{1}{c}{IP} & \multicolumn{1}{c}{NIP} & \multicolumn{1}{c}{IP} \\
    \midrule
    \endhead
    \endfoot
    \midrule
    \endlastfoot
    (30,10) & I & $\alpha_{1}$ & 0.6033& 0.6147& 0.6575& 0.6182& 0.6591& 0.5974& 0.6521& 0.5208& 0.6271& 0.4898& 0.6167 \\
    & & & 0.0528& 0.0509& 0.0067& 0.0514& 0.0066& 0.0505& 0.0070& 0.0634& 0.0097& 0.0730& 0.0112 \\
    & & $\alpha_{2}$ & 1.4077& 1.4341& 1.1582& 1.4424& 1.1586& 1.3923& 1.1568& 1.3384& 1.1537& 1.3052& 1.1521 \\
    & & & 0.3031& 0.3171& 0.0004& 0.3279& 0.0004& 0.2755& 0.0004& 0.2415& 0.0004& 0.2282& 0.0004 \\
    & & $\beta$ & 1.1259& 0.8993& 0.8996& 0.8993& 0.8993& 0.8997& 0.8996& 0.8992& 0.8992& 0.8996& 0.8995 \\
    & & & 0.1558& 0.0004& 0.0003& 0.0004& 0.0003& 0.0004& 0.0003& 0.0004& 0.0003& 0.0004& 0.0003 \\
    & II & $\alpha_{1}$ & 0.6368& 0.6378& 0.6670& 0.6461& 0.6681& 0.6196& 0.6611& 0.5439& 0.6365& 0.5082& 0.6253 \\
    & & & 0.0492& 0.0480& 0.0060& 0.0506& 0.0059& 0.0456& 0.0062& 0.0580& 0.0086& 0.0651& 0.0099 \\
\end{longtable}
   }
   {
   	\scriptsize	
   \begin{longtable}{*{7}c|*{7}c*{10}{r@{}l}}
   	\caption{AL and CP of the $95\%$ confidence and HPD credible intervals of $\alpha_{1}$, $\alpha_{2}$, and $\beta$ for different schemes with different values of $n$, $m$ when $t_{1}=0.4$, $t_{2}=0.6$, and $\Theta=(\alpha_{1},\alpha_{2},\beta)=(0.5,0.8,0.75)$.}
   	\label{T5} \\
   	\toprule
   	\multicolumn{3}{c}{} &
   	\multicolumn{2}{c}{ACI} & \multicolumn{2}{c}{HPD} & \multicolumn{3}{c}{} & \multicolumn{2}{c}{ACI} & \multicolumn{2}{c}{HPD}  \\
   	\cmidrule(lr){4-5} \cmidrule(lr){6-7} \cmidrule(lr){11-12} \cmidrule(lr){13-14}
   	\multicolumn{1}{c}{$(n,m)$} & \multicolumn{1}{c}{CS} & \multicolumn{1}{c}{$\Theta$} & \multicolumn{1}{c}{AL} & \multicolumn{1}{c}{CP} & \multicolumn{1}{c}{AL} & \multicolumn{1}{c}{CP} & \multicolumn{1}{c}{$(n,m)$} & \multicolumn{1}{c}{CS} & \multicolumn{1}{c}{$\Theta$} & \multicolumn{1}{c}{AL} & \multicolumn{1}{c}{CP} & \multicolumn{1}{c}{AL} & \multicolumn{1}{c}{CP} \\
   	\midrule
   	\endfirsthead
   	\caption[]{(continued)} \\
   	\toprule
   	\multicolumn{3}{c}{} &
   	\multicolumn{2}{c}{ACI} & \multicolumn{2}{c}{HPD} & \multicolumn{3}{c}{} & \multicolumn{2}{c}{ACI} & \multicolumn{2}{c}{HPD}  \\
   	\cmidrule(lr){4-5} \cmidrule(lr){6-7} \cmidrule(lr){11-12} \cmidrule(lr){13-14}
   	\multicolumn{1}{c}{$(n,m)$} & \multicolumn{1}{c}{CS} & \multicolumn{1}{c}{$\Theta$} & \multicolumn{1}{c}{AL} & \multicolumn{1}{c}{CP} & \multicolumn{1}{c}{AL} & \multicolumn{1}{c}{CP} & \multicolumn{1}{c}{$(n,m)$} & \multicolumn{1}{c}{CS} & \multicolumn{1}{c}{$\Theta$} & \multicolumn{1}{c}{AL} & \multicolumn{1}{c}{CP} & \multicolumn{1}{c}{AL} & \multicolumn{1}{c}{CP} \\
   	\midrule
   	\endhead
   	\endfoot
   	\midrule
   	\endlastfoot
   	(30,10) & I & $\alpha_{1}$ & 0.934 & 0.964 & 0.600& 0.971& (30,15)& I & $\alpha_{1}$ & 0.860& 0.957& 0.549& 0.964 \\
   	& & $\alpha_{2}$ & 1.712 & 0.978 & 1.198& 0.986& & & $\alpha_{2}$ & 1.228 & 0.974 & 0.870& 0.986 \\
   	& & $\beta$ & 1.362 & 0.758 & 0.412& 0.973& & & $\beta$ & 1.090 & 0.993 & 0.407& 0.953 \\
   	& II & $\alpha_{1}$ & 0.938 & 0.879 & 0.670& 0.947& & II & $\alpha_{1}$ & 0.899 & 0.924 & 0.643& 0.957 \\
   	& & $\alpha_{2}$ & 1.538 & 0.922 & 1.228& 0.962& & & $\alpha_{2}$ & 1.300 & 0.948 & 0.869& 0.961 \\
   	& & $\beta$ & 1.090 & 0.993 & 0.410& 0.995& & & $\beta$ & 0.801 & 0.869 & 0.406& 0.955 \\
   	& III & $\alpha_{1}$ & 0.938 & 0.957 & 0.610& 0.964& & III & $\alpha_{1}$ & 0.884 & 0.923 & 0.643& 0.958 \\
   	& & $\alpha_{2}$ & 1.665 & 0.973 & 1.209& 0.977& & & $\alpha_{2}$ & 1.284 & 0.944 & 0.869& 0.975 \\
\end{longtable}}
   {
   	\scriptsize
   \begin{longtable}{*{7}c|*{7}c*{10}{r@{}l}}
   	\caption{AL and CP of the $95\%$ confidence and HPD credible intervals of $\alpha_{1}$, $\alpha_{2}$, and $\beta$ for different schemes with different values of $n$, $m$ when $t_{1}=0.4$, $t_{2}=0.6$, and $\Theta=(\alpha_{1},\alpha_{2},\beta)=(0.7,1.15,0.9)$.}
   	\label{T6} \\
   	\toprule
   	\multicolumn{3}{c}{} &
   	\multicolumn{2}{c}{ACI} & \multicolumn{2}{c}{HPD} & \multicolumn{3}{c}{} & \multicolumn{2}{c}{ACI} & \multicolumn{2}{c}{HPD}  \\
   	\cmidrule(lr){4-5} \cmidrule(lr){6-7} \cmidrule(lr){11-12} \cmidrule(lr){13-14}
   	\multicolumn{1}{c}{$(n,m)$} & \multicolumn{1}{c}{CS} & \multicolumn{1}{c}{$\Theta$} & \multicolumn{1}{c}{AL} & \multicolumn{1}{c}{CP} & \multicolumn{1}{c}{AL} & \multicolumn{1}{c}{CP} & \multicolumn{1}{c}{$(n,m)$} & \multicolumn{1}{c}{CS} & \multicolumn{1}{c}{$\Theta$} & \multicolumn{1}{c}{AL} & \multicolumn{1}{c}{CP} & \multicolumn{1}{c}{AL} & \multicolumn{1}{c}{CP} \\
   	\midrule
   	\endfirsthead
   	\caption[]{(continued)} \\
   	\toprule
   	\multicolumn{3}{c}{} &
   	\multicolumn{2}{c}{ACI} & \multicolumn{2}{c}{HPD} & \multicolumn{3}{c}{} & \multicolumn{2}{c}{ACI} & \multicolumn{2}{c}{HPD}  \\
   	\cmidrule(lr){4-5} \cmidrule(lr){6-7} \cmidrule(lr){11-12} \cmidrule(lr){13-14}
   	\multicolumn{1}{c}{$(n,m)$} & \multicolumn{1}{c}{CS} & \multicolumn{1}{c}{$\Theta$} & \multicolumn{1}{c}{AL} & \multicolumn{1}{c}{CP} & \multicolumn{1}{c}{AL} & \multicolumn{1}{c}{CP} & \multicolumn{1}{c}{$(n,m)$} & \multicolumn{1}{c}{CS} & \multicolumn{1}{c}{$\Theta$} & \multicolumn{1}{c}{AL} & \multicolumn{1}{c}{CP} & \multicolumn{1}{c}{AL} & \multicolumn{1}{c}{CP} \\
   	\midrule
   	\endhead
   	\endfoot
   	\midrule
   	\endlastfoot
   	(30,10) & I & $\alpha_{1}$ & 1.363 & 0.976 & 0.772& 0.982& $(30,15)$ & I & $\alpha_{1}$ & 1.242 & 0.985 & 0.704& 0.987 \\
   	& & $\alpha_{2}$ & 2.564 & 0.995 & 0.789& 0.997& & & $\alpha_{2}$ & 1.877 & 0.991 & 0.527& 0.993 \\
   	& & $\beta$ & 1.303 & 0.739 & 0.394& 0.975& & & $\beta$ & 1.055 & 0.755 & 0.392& 0.976 \\
   	& II & $\alpha_{1}$ & 1.379 & 0.902 & 0.800& 0.948& & II & $\alpha_{1}$ & 1.249 & 0.886 & 0.837& 0.965 \\
   	& & $\alpha_{2}$ & 2.384 & 0.909 &  0.837& 0.962& & & $\alpha_{2}$ & 1.847 & 0.892 & 0.646& 0.953 \\
   	& & $\beta$ & 1.220 & 0.959 & 0.394& 0.973& & & $\beta$ & 0.902 & 0.847 & 0.392& 0.968 \\
   	& III & $\alpha_{1}$ & 1.362 & 0.976 & 0.766& 0.981& & III & $\alpha_{1}$ & 1.257 & 0.883 & 0.838& 0.962  \\
   	& & $\alpha_{2}$ & 2.530 & 0.987 & 0.793& 0.988& & & $\alpha_{2}$ & 1.854 & 0.891 & 0.649& 0.975 \\
\end{longtable}}
   {
   	\scriptsize
   \begin{longtable}{*{7}c|*{7}c*{10}{r@{}l}}
   	\caption{AL and CP of the $95\%$ confidence and HPD credible intervals of $\alpha_{1}$, $\alpha_{2}$, and $\beta$ for different schemes with different values of $n$, $m$ when $t_{1}=0.8$, $t_{2}=1.12$, and $\Theta=(\alpha_{1},\alpha_{2},\beta)=(0.5,0.8,0.75)$.}
   	\label{T7} \\
   	\toprule
   	\multicolumn{3}{c}{} &
   	\multicolumn{2}{c}{ACI} & \multicolumn{2}{c}{HPD} & \multicolumn{3}{c}{} & \multicolumn{2}{c}{ACI} & \multicolumn{2}{c}{HPD}  \\
   	\cmidrule(lr){4-5} \cmidrule(lr){6-7} \cmidrule(lr){11-12} \cmidrule(lr){13-14}
   	\multicolumn{1}{c}{$(n,m)$} & \multicolumn{1}{c}{CS} & \multicolumn{1}{c}{$\Theta$} & \multicolumn{1}{c}{AL} & \multicolumn{1}{c}{CP} & \multicolumn{1}{c}{AL} & \multicolumn{1}{c}{CP} & \multicolumn{1}{c}{$(n,m)$} & \multicolumn{1}{c}{CS} & \multicolumn{1}{c}{$\Theta$} & \multicolumn{1}{c}{AL} & \multicolumn{1}{c}{CP} & \multicolumn{1}{c}{AL} & \multicolumn{1}{c}{CP} \\
   	\midrule
   	\endfirsthead
   	\caption[]{(continued)} \\
   	\toprule
   	\multicolumn{3}{c}{} &
   	\multicolumn{2}{c}{ACI} & \multicolumn{2}{c}{HPD} & \multicolumn{3}{c}{} & \multicolumn{2}{c}{ACI} & \multicolumn{2}{c}{HPD}  \\
   	\cmidrule(lr){4-5} \cmidrule(lr){6-7} \cmidrule(lr){11-12} \cmidrule(lr){13-14}
   	\multicolumn{1}{c}{$(n,m)$} & \multicolumn{1}{c}{CS} & \multicolumn{1}{c}{$\Theta$} & \multicolumn{1}{c}{AL} & \multicolumn{1}{c}{CP} & \multicolumn{1}{c}{AL} & \multicolumn{1}{c}{CP} & \multicolumn{1}{c}{$(n,m)$} & \multicolumn{1}{c}{CS} & \multicolumn{1}{c}{$\Theta$} & \multicolumn{1}{c}{AL} & \multicolumn{1}{c}{CP} & \multicolumn{1}{c}{AL} & \multicolumn{1}{c}{CP} \\
   	\midrule
   	\endhead
   	\endfoot
   	\midrule
   	\endlastfoot
   	(30,10) & I & $\alpha_{1}$ & 0.932 & 0.953 & 0.597& 0.962& $(30,15)$ & I & $\alpha_{1}$ & 0.827 & 0.957 & 0.539& 0.963 \\
   	& & $\alpha_{2}$ & 1.717 & 0.989 & 1.164& 0.995& & & $\alpha_{2}$ & 1.239 & 0.975 & 0.843& 0.978 \\
   	& & $\beta$ & 1.328 & 0.988 & 0.411& 0.996& & & $\beta$ & 1.024 & 0.978 & 0.407& 0.998 \\
   	& II & $\alpha_{1}$ & 0.918 & 0.943 & 0.589& 0.953& & II & $\alpha_{1}$ & 0.844 & 0.914 & 0.572& 0.946 \\
   	& & $\alpha_{2}$ & 1.625 & 0.976 & 1.165& 0.982& & & $\alpha_{2}$ & 1.190 & 0.912 & 0.904& 0.962 \\
   	& & $\beta$ & 1.287 & 0.858 & 0.410& 0.966& & & $\beta$ & 0.948 & 0.974 & 0.406& 0.985 \\
   	& III & $\alpha_{1}$ & 0.916 & 0.945 & 0.598& 0.952& & III & $\alpha_{1}$ & 0.842 & 0.910 & 0.572& 0.966 \\
   	& & $\alpha_{2}$ & 1.672 & 0.985 & 1.162& 0.987& & & $\alpha_{2}$ & 1.188 & 0.909 & 0.903& 0.953 \\
\end{longtable}}
   {
   	\scriptsize
   \begin{longtable}{*{7}c|*{7}c*{10}{r@{}l}}
   	\caption{AL and CP of the $95\%$ confidence and HPD credible intervals of $\alpha_{1}$, $\alpha_{2}$, and $\beta$ for different schemes with different values of $n$, $m$ when $t_{1}=0.8$, $t_{2}=1.12$, and $\Theta=(\alpha_{1},\alpha_{2},\beta)=(0.7,1.15,0.9)$.}
   	\label{T8} \\
   	\toprule
   	\multicolumn{3}{c}{} &
   	\multicolumn{2}{c}{ACI} & \multicolumn{2}{c}{HPD} & \multicolumn{3}{c}{} & \multicolumn{2}{c}{ACI} & \multicolumn{2}{c}{HPD}  \\
   	\cmidrule(lr){4-5} \cmidrule(lr){6-7} \cmidrule(lr){11-12} \cmidrule(lr){13-14}
   	\multicolumn{1}{c}{$(n,m)$} & \multicolumn{1}{c}{CS} & \multicolumn{1}{c}{$\Theta$} & \multicolumn{1}{c}{AL} & \multicolumn{1}{c}{CP} & \multicolumn{1}{c}{AL} & \multicolumn{1}{c}{CP} & \multicolumn{1}{c}{$(n,m)$} & \multicolumn{1}{c}{CS} & \multicolumn{1}{c}{$\Theta$} & \multicolumn{1}{c}{AL} & \multicolumn{1}{c}{CP} & \multicolumn{1}{c}{AL} & \multicolumn{1}{c}{CP} \\
   	\midrule
   	\endfirsthead
   	\caption[]{(continued)} \\
   	\toprule
   	\multicolumn{3}{c}{} &
   	\multicolumn{2}{c}{ACI} & \multicolumn{2}{c}{HPD} & \multicolumn{3}{c}{} & \multicolumn{2}{c}{ACI} & \multicolumn{2}{c}{HPD}  \\
   	\cmidrule(lr){4-5} \cmidrule(lr){6-7} \cmidrule(lr){11-12} \cmidrule(lr){13-14}
   	\multicolumn{1}{c}{$(n,m)$} & \multicolumn{1}{c}{CS} & \multicolumn{1}{c}{$\Theta$} & \multicolumn{1}{c}{AL} & \multicolumn{1}{c}{CP} & \multicolumn{1}{c}{AL} & \multicolumn{1}{c}{CP} & \multicolumn{1}{c}{$(n,m)$} & \multicolumn{1}{c}{CS} & \multicolumn{1}{c}{$\Theta$} & \multicolumn{1}{c}{AL} & \multicolumn{1}{c}{CP} & \multicolumn{1}{c}{AL} & \multicolumn{1}{c}{CP} \\
   	\midrule
   	\endhead
   	\endfoot
   	\midrule
   	\endlastfoot
   	(30,10) & I & $\alpha_{1}$ & 1.365 & 0.975 & 0.775& 0.982& $(30,15)$ & I & $\alpha_{1}$ & 1.253 & 0.986 & 0.705& 0.987 \\
   	& & $\alpha_{2}$ & 2.560 & 0.993 & 0.787& 0.996& & & $\alpha_{2}$ & 1.903& 0.994& 0.535& 0.994 \\
   	& & $\beta$ & 1.303 & 0.743 & 0.394& 0.953& & & $\beta$ & 1.015 & 0.994 & 0.391& 0.995  \\
   	& II & $\alpha_{1}$ & 1.357& 0.968& 0.774& 0.976& & II & $\alpha_{1}$ & 1.244& 0.978& 0.707& 0.984 \\
   	& & $\alpha_{2}$ & 2.488 & 0.993 & 0.790& 0.995& & & $\alpha_{2}$ & 1.845& 0.982& 0.535& 0.987 \\
   	& & $\beta$ & 1.305 & 0.792 & 0.393& 0.951& & & $\beta$ & 1.003 & 0.862 & 0.389& 0.975 \\
   	& III & $\alpha_{1}$ & 1.358 & 0.972 & 0.775& 0.982& & III & $\alpha_{1}$ & 1.238 & 0.978 & 0.705 & 0.983 \\
   	& & $\alpha_{2}$ & 2.530 & 0.993 & 0.789& 0.995& & & $\alpha_{2}$ & 1.835 & 0.982 & 0.533& 0.991 \\
\end{longtable}}
   \section{Optimality criterion}
   In survival analysis, selecting an optimal censoring plan from chosen schemes is often essential to obtain adequate information about the unknown model parameters. In this study, we have employed three widely used criteria based on the variance-covariance matrix (VCM) of the observed Fisher information matrix associated with the MLEs of the unknown parameters. In statistical literature, the $A$ and $D$-optimality criteria are the most commonly utilized. These two criteria minimize the trace and determinant of the VCM, respectively. Additionally, $F$-optimality aims to minimize the trace of the observed Fisher information matrix for the MLEs. Based on these criteria, the corresponding optimal censoring scheme is presented in Table \ref{T9}.
   \begin{table}[htbp!]
   	\begin{center}
   	\caption{Different optimality criterion.}
   	\label{T9}
   	\tabcolsep 7pt
   	 \begin{tabular}{*{3}c*{2}{r@{}l}}
   	 \toprule
   	 \multicolumn{1}{c}{Criterion}&&& & &\multicolumn{1}{c}{Goal}  \\
   	 \midrule
   	 A-optimality &&& & & minimum trace $(I^{-1}(\widehat{\Theta}))$&\\
   	 D-optimality &&& & &minimum det $(I^{-1}(\widehat{\Theta}))$~~~& \\
   	 F-optimality &&&&& maximum  trace $(I(\widehat{\Theta}))$~~~&\\				
   	 \bottomrule
   	 \end{tabular}
   	\end{center}
   \end{table}
   \section{Analyzing real data}
   The importance of the theoretical findings that were discussed in the previous parts will be described in this section using an example from a research survey on affects of food availability to certain bird population. This analysis of the real-world data set confirms the validity of the proposed point and interval estimates for the unknown parameters.

   In this application, we analyze the data provided by Briga et al. \cite{briga2017food}. According to the data, chicks of certain species called Zebra Finches were reared from birth upto $35$ or $120$ days and adult birds were housed for lifetime in different environmental conditions based on food availability by manipulating foraging costs. For clarity, we state four experimental groups as BB, BH, HB, HH, where first letter indicates benign (B) or harsh (H) developmental conditions, and the second letter indicates for benign (B) or harsh (H) foraging conditions in adulthood. For simplicity, here we consider only two conditions such as BB for cause $1$ and HB for cause $2$.
   
   Assuming independent Gompertz distributions for the latent causes of failure, the Kolmogorov-Smirnov (K-S) test is used to assess the goodness of fit of the proposed model to the two causes of failure, under the hypotheses $H_{0}$ (data follows the distribution) and $H_{1}$ (data does not follow the distribution). 
   \begin{table}[htbp!]
   	\begin{center}
    \caption{The K-S test statistic result for Briga $(2016)$ data.}
   	\label{T10}
   	\tabcolsep 10pt
   	 \begin{tabular}{*{11}c*{10}{r@{}l}}
   	 \toprule
   	 \multicolumn{1}{c}{Dist.} & \multicolumn{1}{c}{Data} & \multicolumn{1}{c}{$\hat{\theta}$} & \multicolumn{1}{c}{K-S statistic} & \multicolumn{1}{c}{$p$-value} & \multicolumn{1}{c}{AIC} \\
   	 \midrule
   	 Exponential& Cause $1$& $(\hat{\alpha}=0.3724)$& $0.1675$& $0.0014$& $514.85$\\[1mm]
   	 & Cause $2$& $(\hat{\alpha}=0.3966)$& $0.1416$& $0.0096$& $513.99$\\
   	 \midrule
   	 Weibull& Cause $1$& $(\hat{\alpha}=1.5090, \hat{\beta}=2.9629)$& $0.0859$& $0.2958$& $489.31$\\[1mm]
   	 & Cause $2$& $(\hat{\alpha}=1.3442, \hat{\beta}=2.7360)$& $0.0752$& $0.4394$& $500.41$\\
   	 \midrule
   	 Gompertz& Cause $1$& $(\hat{\alpha}=0.1864, \hat{\beta}=0.3045)$& $0.0559$& $0.8149$& $484.30$\\[1mm]
   	 & Cause $2$& $(\hat{\alpha}=0.2378, \hat{\beta}=0.2329)$& $0.0511$& $0.8779$& $495.81$\\
   	 \midrule
   	 \end{tabular}
   	\end{center}
   \end{table}

   Depending on the results from Table \ref{T10}, for the two causes of failure, we can say that at $5\%$ significance level; the computed value of the K-S statistics are lower than that of the other distributions, such as exponential, Weibull. Also, the p-values $(0.81)$, $(0.87)$ for cause $1$ and cause $2$ from the K-S test for Gompertz distribution are larger than the significance level $(0.05)$ which indicates that the Gompertz model generally fits well with the real data rather than exponential and Weibull model. Additionally, empirical cumulative distribution function (ECDF) plots, probability-probability (P-P) plots and quantile-quantile (Q-Q) plots are provided in Figure \ref{fig:fig5} which support that Gompertz distribution is suitable to fit this data.

   Applying various censoring schemes, we have generated two IAT-II progressively censored competing risks failure data sets with $N=31$, $N=33$ for $t_{0}=0.7$ and $t_{0}=0.85$, respectively. This competing risks data is provided in Table \ref{T11}. Furthermore, the equivalence between the shape and scale parameters are tested in Section $5$ for the real dataset with $5\%$ level of significance. We have obtained that the LRS value for $\alpha$ and $\beta$ with the corresponding $p$-values (within brackets) as $6.852(0.009)$ and $0.679(0.410)$ for Data I (From Table \ref{T11}) respectively and LRS value for $\alpha$ and $\beta$ with the corresponding $p$-values (within brackets) as $5.031 (0.024)$ and $0.884 (0.347)$ for Data II (From Table \ref{T11}) respectively. From this results, it can be decided that the null hypothesis for test $(1)$ is rejected but that for test $(2)$ is accepted. Therefore, it may be assumed that the scale parameters $(\beta)$ are equal and shape parameters $(\alpha)$ are different.
   
   Depending on the competing risks data provided in Table \ref{T11}, point and interval estimates have been calculated and shown in Table \ref{T12}. For interval estimates, a significance level $(\gamma)$ of 0.05 is considered and the interval lengths are presented in square brackets. To calculate the Bayes estimates, NIP has been implemented under SELF. Table \ref{T12} shows that the Bayes and point estimates based on MLEs are nearly identical. In comparison to ACIs, the HPD credible intervals perform better. Further, the profile log-likelihood of the unknown model parameters are plotted in Figure \ref{fig:fig4}, which suggests that the MLEs of $\alpha_{1}$, $\alpha_{2}$ and $\beta$ are unique. Additionally, we have computed three distinct optimality criteria for two data sets using censoring schemes I, II, and III. Depending upon these criteria, it is to be observed that CS-II and CS-III is optimal than CS-I for both of the data sets. 
   \newpage
   \begin{table}[ht!]
   	\begin{center}
   	\caption{ IAT-II PCS competing risks data set with two causes of failures generated from real data.}
   	\label{T11}
   	\tabcolsep 7pt
   	\small
   	\scalebox{0.9}{
   	 \begin{tabular}{*{1}c*{1}{r@{}l}}
   	 \toprule
   	 \textbf{Data I :}  $t_{0}=0.7$ and $N=31$\\
   	 \midrule
   	 (0.01,2),(0.04,1),(0.08,1),(0.10,2),(0.13,1),(0.16,2),(0.18,2),(0.19,1),(0.20,1),(0.23,2),(0.31,1),(0.32,1),\\(0.33,1),(0.34,1),(0.36,2),(0.38,1),(0.45,1),(0.46,1),(0.50,1),(0.52,2),(0.53,2),(0.54,1),(0.54,2),(0.58,1),\\(0.58,2),(0.60,1),(0.60,2),(0.66,1),(0.66,2),(0.67,1),(0.68,2).~~~~~~~~~~~~~~~~~~~~~~~~~~~~~~~~~~~~~~~~~~~~~~~~~~~~~~~\\
   	 \midrule
   	 \textbf{Data II :} $t_{0}=0.85$ and $N=33$\\
   	 \midrule
   	 (0.01,2),(0.04,1),(0.08,2),(0.10,1),(0.13,2),(0.16,1),(0.18,1),(0.19,1),(0.20,2),(0.23,1),(0.31,1),(0.31,2),\\(0.32,1),(0.33,1),(0.33,2),(0.34,2),(0.36,2),(0.38,2),(0.45,1),(0.46,1),(0.50,1),(0.52,1),(0.53,1),(0.54,1),\\(0.58,1),(0.58,2),(0.60,2),(0.66,2),(0.67,2),(0.68,2),(0.70,2),(0.78,2),(0.81,2).~~~~~~~~~~~~~~~~~~~~~~~~~~~~~~~~~~\\
   	 \bottomrule
   	 \end{tabular}}
   	\end{center}
   \end{table}
   {
   \scriptsize
   \begin{longtable}{*{9}c}
   	\caption{Point and interval estimates based on real data.}
   	\label{T12}\\
   	\midrule
   	\multicolumn{1}{c}{Data} & \multicolumn{1}{c}{CS} & \multicolumn{1}{c}{$\Theta$} & \multicolumn{1}{c}{$\alpha_{1}$} & \multicolumn{1}{c}{$\alpha_{2}$} & \multicolumn{1}{c}{$\beta$} & \multicolumn{1}{c}{A-optimality} & \multicolumn{1}{c}{D-optimality} & \multicolumn{1}{c}{F-optimality} \\
   	\midrule
   	\endfirsthead
   	\caption[]{(continued)} \\
   	\midrule
   	\multicolumn{1}{c}{Data} & \multicolumn{1}{c}{CS} & \multicolumn{1}{c}{$\Theta$} & \multicolumn{1}{c}{$\alpha_{1}$} & \multicolumn{1}{c}{$\alpha_{2}$} & \multicolumn{1}{c}{$\beta$} & \multicolumn{1}{c}{A-optimality} & \multicolumn{1}{c}{D-optimality} & \multicolumn{1}{c}{F-optimality} \\
   	\midrule
   	\endhead
   	\endfoot
   	\midrule
   	\endlastfoot
   	I & I & MLE & 0.05955& 0.05953& 2.02290& 0.328007& 8.499$\times 10^{-9}$& 1.242$\times 10^{4}$ \\
   	& & Bayes & 0.05999& 0.06016& 2.00754& \\
   	& & ACI & (0.0233,0.0958)& (0.0233,0.0958)& (0.902,3.144)& \\
   	& & & [0.07254]& [0.07250]& [2.24271]& \\
   	& & HPD & (0.0369,0.0854)& (0.0345,0.0831)& (2,2.03)& \\
   	& & & [0.04851]& [0.04866]& [0.03095]& \\[2 mm]
   	& II & MLE & 0.05594& 0.05608& 2.33103& 0.2739495& 5.582$\times 10^{-9}$& 1.4$\times 10^{4}$ \\
   	& & Bayes & 0.05605& 0.05550& 2.33775& \\
   	& & ACI & (0.0241,0.088)& (0.0241,0.088)& (1.31,3.36)& \\
   	& & & [0.06399]& [0.06404]& [2.04974]& \\
   	& & HPD & (0.0325,0.0791)& (0.0329,0.0782)& (2.32,2.36)& \\
   	& & & [0.046649]& [0.04525]& [0.042961]& \\[2 mm]
   	& III & MLE & 0.05605& 0.05608& 2.33103& 0.2739495& 5.582$\times 10^{-9}$& 1.4$\times 10^{4}$ \\
   	& & Bayes & 0.05623& 0.05606& 2.32454& \\
   	& & ACI & (0.0241,0.088)& (0.0241,0.0881)& (1.31,3.36)& \\
   	& & & [0.06399]& [0.06404]& [2.04974]& \\
   	& & HPD & (0.0332,0.0791)& (0.0363,0.0812)& (2.3,2.34)& \\
   	& & & [0.04588]& [0.04481]& [0.03871]& \\ [2 mm]
   	II & I & MLE & 0.08804& 0.08807& 0.68029& 0.04934039& 5.074$\times 10^{-9}$& 6216.268 \\
   	& & Bayes & 0.08981& 0.08903& 0.66499& \\
   	& & ACI & (0.0506,0.1255)& (0.0506,0.1255)& (0.248,1.112)& \\
   	& & & [0.07487]& [0.07490]& [0.86427]& \\
   	& & HPD & (0.0537,0.1258)& (0.0546,0.1277)& (0.627,0.703)& \\
   	& & & [0.07208]& [0.07308]& [0.07582]& \\ [2 mm]
   	& II & MLE & 0.08229& 0.08235& 0.99344& 0.05427762& 4.277$\times 10^{-9}$& 7105.881 \\
   	& & Bayes & 0.08144& 0.08208& 1.01641& \\
   	& & ACI & (0.0475,0.1171)& (0.0475,0.1172)& (0.539,1.447)& \\
   	& & & [0.06957]& [0.06963]& [0.90794]& \\
   	& & HPD & (0.0512,0.1162)& (0.0479,0.1119)& (0.984,1.051)& \\
   	& & & [0.06497]& [0.06398]& [0.06702]& \\ [2 mm]
   	& III & MLE & 0.08229& 0.08235& 0.99344& 0.05427762& 4.277$\times 10^{-9}$& 7105.881 \\
   	& & Bayes & 0.08314& 0.08301& 0.98484& \\
   	& & ACI & (0.0475,0.1171)& (0.0475,0.1172)& (0.539,1.447)& \\
   	& & & [0.06957]& [0.06963]& [0.90794]& \\
   	& & HPD & (0.0513,0.115)& (0.0509,0.1149)& (0.966,1.003)& \\
   	& & & [0.06367]& [0.06404]& [0.03612]& \\ [2 mm] 
   \end{longtable}}
   \begin{figure}[h!]
   	\centering
   	\subfigure[]{\includegraphics[width=0.32\textwidth,height=0.31\textwidth]{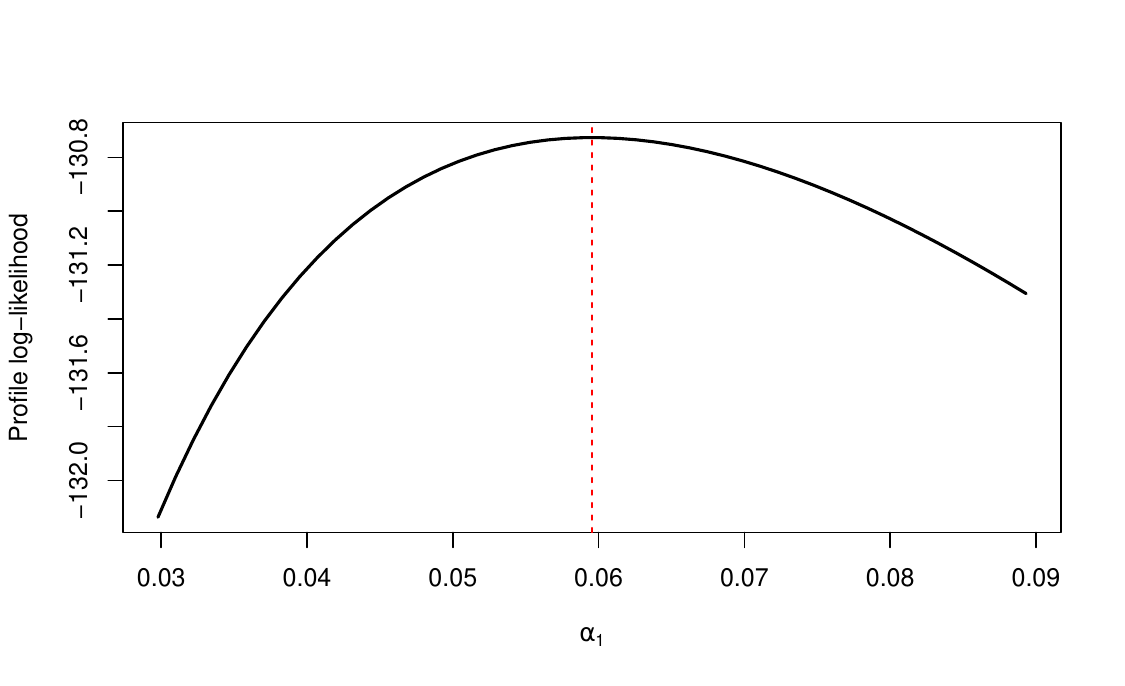}}
   	\subfigure[]{\includegraphics[width=0.32\textwidth,height=0.31\textwidth]{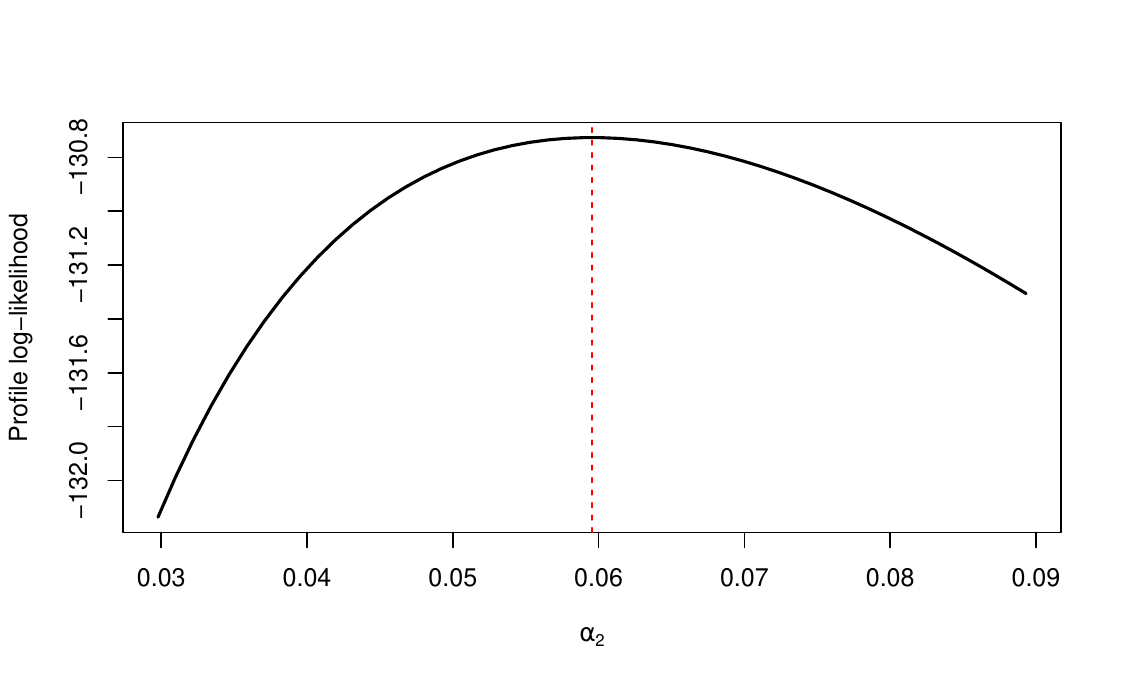}}
   	\subfigure[]{\includegraphics[width=0.32\textwidth,height=0.31\textwidth]{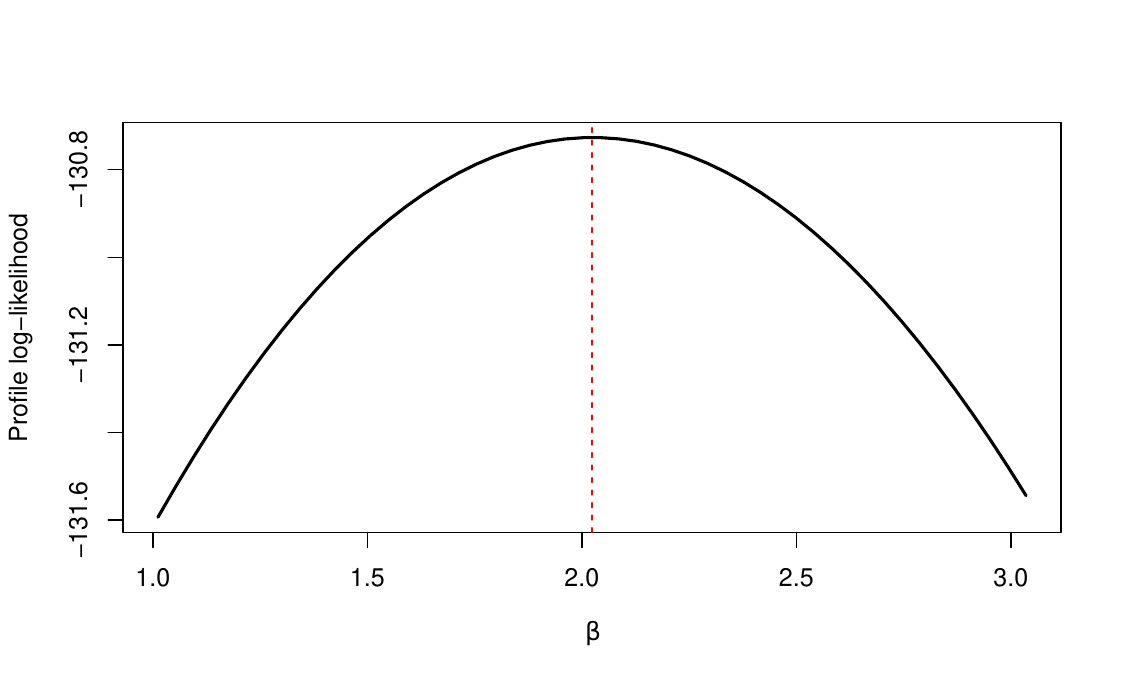}}
   	\subfigure[]{\includegraphics[width=0.32\textwidth,height=0.31\textwidth]{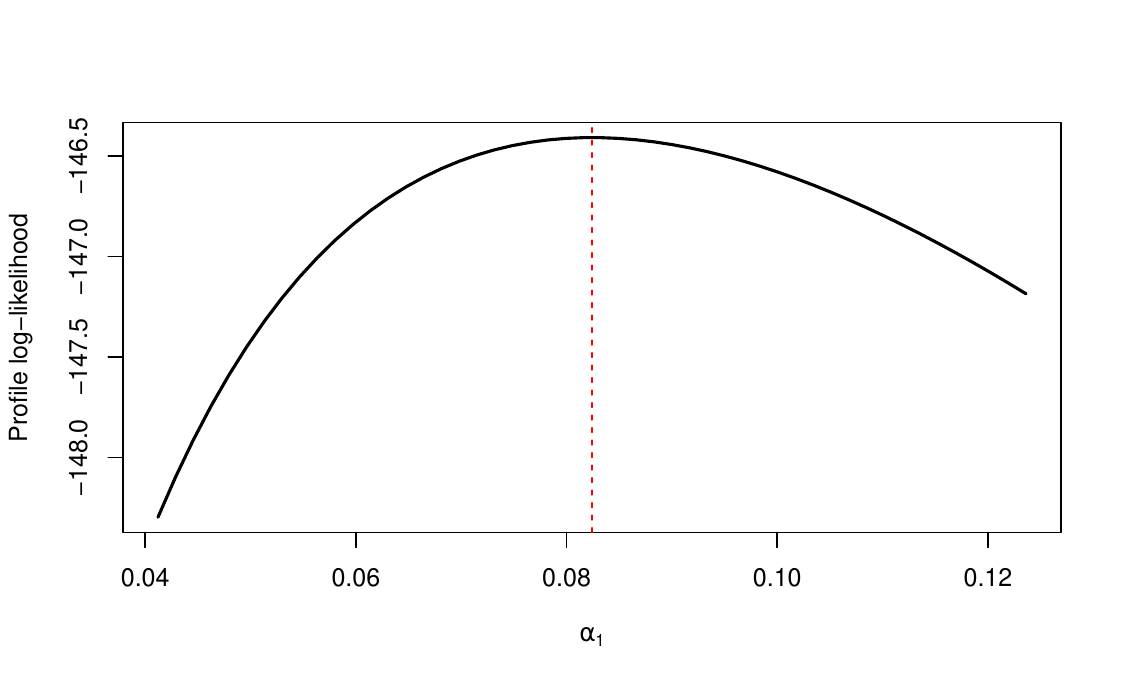}}
   	\subfigure[]{\includegraphics[width=0.32\textwidth,height=0.31\textwidth]{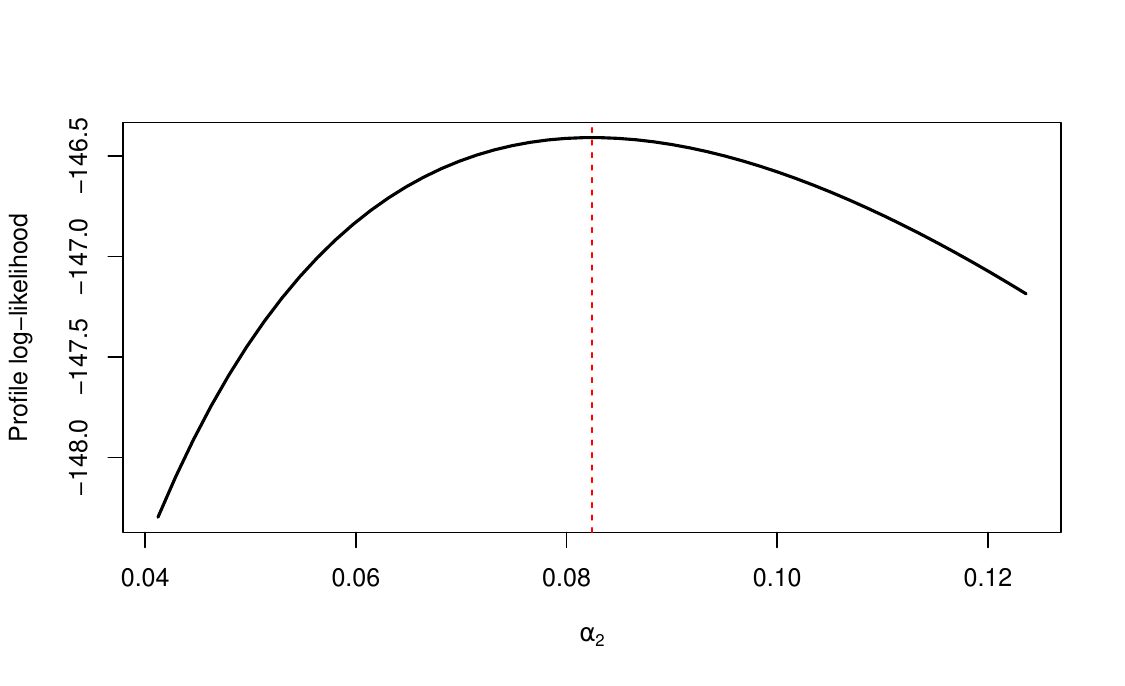}}
   	\subfigure[]{\includegraphics[width=0.32\textwidth,height=0.31\textwidth]{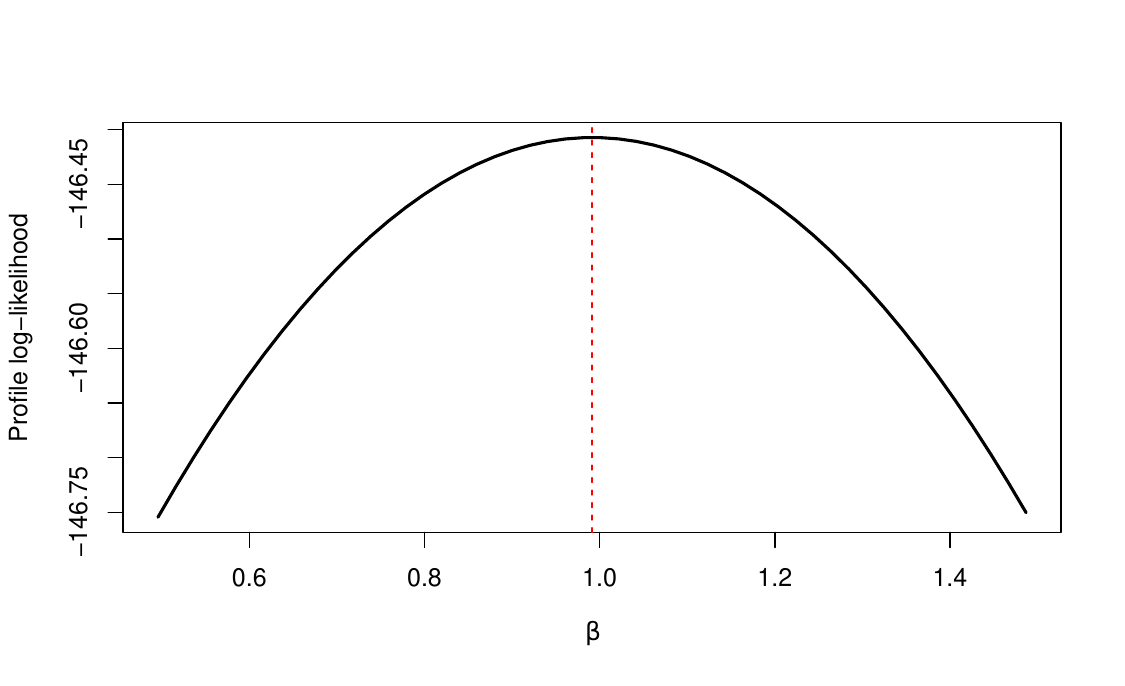}}
   	\caption{Profile log-likelihood plot for (a) $\alpha_{1}$, (b) $\alpha_{2}$, (c) $\beta$ under CS-I for Data I; profile log-likelihood plot for (d) $\alpha_{1}$, (e) $\alpha_{2}$, (f) $\beta$ under CS-III for Data-II.}
   	\label{fig:fig4}
   \end{figure}
   \begin{figure}[h!]
   	\centering
   	\subfigure[]{\includegraphics[width=0.32\textwidth,height=0.31\textwidth]{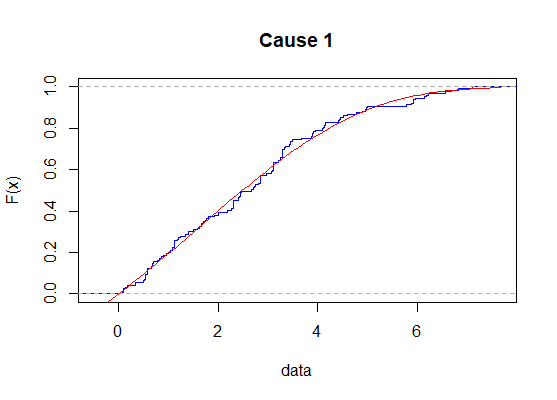}} 
   	\subfigure[]{\includegraphics[width=0.32\textwidth,height=0.31\textwidth]{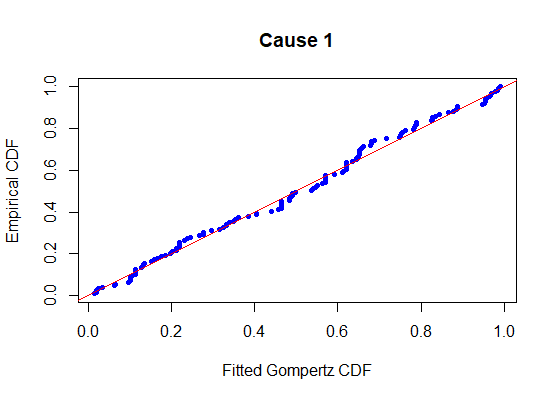}} 
   	\subfigure[]{\includegraphics[width=0.32\textwidth,height=0.31\textwidth]{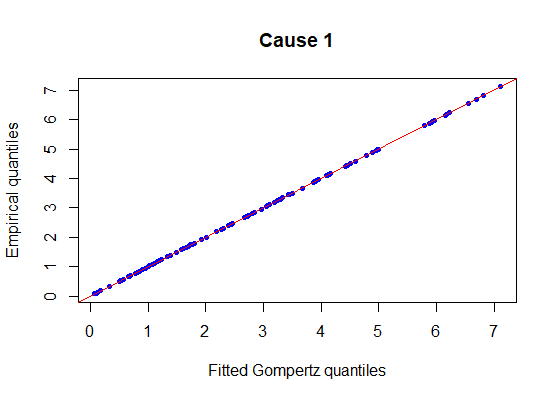}}
   	\subfigure[]{\includegraphics[width=0.32\textwidth,height=0.31\textwidth]{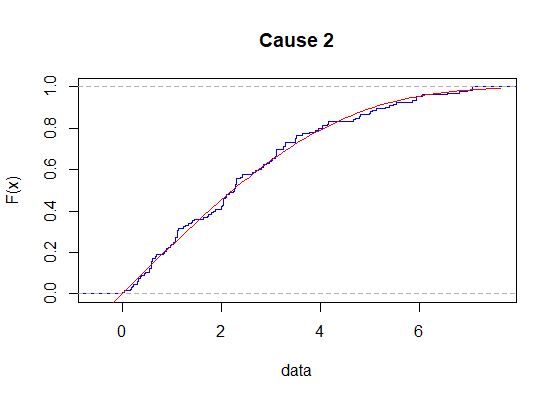}}
   	\subfigure[]{\includegraphics[width=0.32\textwidth,height=0.31\textwidth]{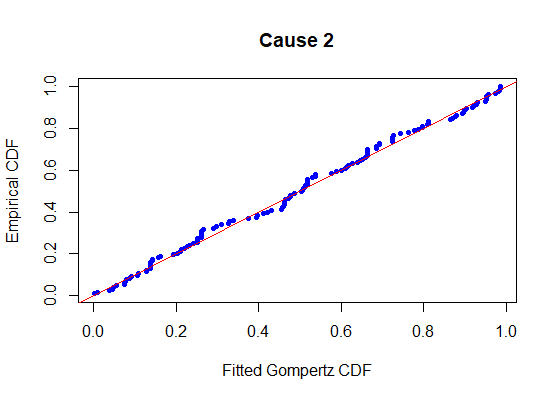}}
   	\subfigure[]{\includegraphics[width=0.32\textwidth,height=0.31\textwidth]{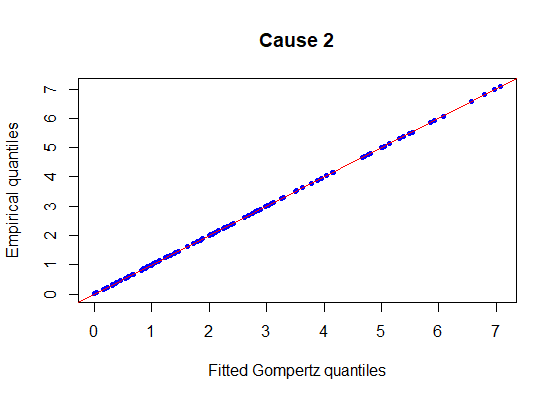}}
   	\caption{ECDF, P-P and Q-Q plot for Briga' data: (a) ECDF under cause 1, (b) P-P plot under cause 1, (c) Q-Q plot under cause 1, (d) ECDF under cause 2, (e) P-P plot under cause 2, (f) Q-Q plot under cause 2.}
   	\label{fig:fig5}
   \end{figure}
   
   \section{Conclusion}
   In the presence of competing risks data, statistical inferences have been made in this study using both classical and Bayesian approaches under IAT-II PCS. Here, we have taken two causes of failures. It is assumed that the latent failure times due to two different and independent causes are following Gompertz lifetime distribution with different shape and scale parameters. The MLEs have been derived and it is established that they exist uniquely. Further, We have derived confidence intervals for the unknown parameters using the asymptotic normality property of MLEs. Also, we have obtained Bayes estimates under both non-informative and informative priors. Three loss functions (SELF, LLF, GELF) have been used. We have obtained HPD credible intervals from posterior density functions. Three different optimality criteria such as $A$, $D$ and $F$ have been calculated to obtain the optimal censoring plan. In simulation studies, it is observed that Bayes estimates perform better than MLEs and the Bayes estimates with respect to IP perform better Bayes estimates with respect to NIP. Also, the HPD credible intervals perform better than asymptotic confidence intervals in terms of AL and CP. To demonstrate this phenomena, we have also examined a real-life data set for for the purpose of illustrations.
   \section*{Conflicts of interest}
   The authors declare no conflict of interest.
\bibliography{ref1} 
\end{document}